%% file: main.tex
\documentclass{article}

\input{settings/packages}
\input{settings/macros}
\iclrfinalcopy

\title{MIRT: Transformers for Truthful Generative Auctions with Whole-feed Permutation Externalities}

\author{%
    \parbox[t]{0.95\textwidth}{\centering\normalfont
      {\small\bfseries
        Ali Elahi$^{1}$\enspace
        Ermis Soumalias$^{2}$\enspace
        Jason Cheuk Nam Liang$^{2}$\enspace
        Daniel Yao$^{2}$\enspace
        Michael J. Curry$^{1}$
      }\\[0.6em]
      {\small $^{1}$University of Illinois Chicago \qquad $^{2}$Meta}\\[0.6em]
      {\footnotesize $^{1}$\texttt{\{aelahi6, mjc\}@uic.edu}\\[0.2em] $^{2}$\texttt{\{esoumalias, jcnliang, danielyao\}@meta.com}}
    }%
  }

\begin{document}

\maketitle

\fancyhead{}
\lhead{Preprint.}
% Abstract
\input{sections/00-abstract}

% Introduction
\input{sections/10-introduction}

% Our Model
\input{sections/20-model}

\input{sections/30-auctions}

% Experimental Results
\input{sections/40-experiments}

% Conclusion and Future Directions
\input{sections/50-conclusion}

\bibliographystyle{iclr2027_conference}
\bibliography{references}

%%%%%%%%%%%%%%%%%%%%%%%%%%%%%%%%%%%%%%%%%%%%%%%%%%%%%%%%%%%%

% \appendix

% \section{Technical appendices and supplementary material}
% Technical appendices with additional results, figures, graphs, and proofs may be submitted with the paper submission before the full submission deadline (see above). You can upload a ZIP file for videos or code, but do not upload a separate PDF file for the appendix. There is no page limit for the technical appendices. 

% Note: Think of the appendix as ``optional reading'' for reviewers. The paper must be able to stand alone without the appendix; for example, adding critical experiments that support the main claims to an appendix is inappropriate. 

% %%%%%%%%%%%%%%%%%%%%%%%%%%%%%%%%%%%%%%%%%%%%%%%%%%%%%%%%%%%%
\clearpage
\input{sections/60-appendix}
% \newpage
% \input{sections/100 - checklist}

\end{document}

%% file: settings/packages.tex
\usepackage{iclr2027_conference,times}
\usepackage[utf8]{inputenc}
\usepackage[T1]{fontenc}

\usepackage{
  amsmath,
  amssymb,
  amsfonts,
  mathtools,
  amsthm
}

\usepackage{thm-restate}

\usepackage{bbm}
\usepackage{booktabs}
\usepackage{multirow}
\usepackage{colortbl}
\usepackage{etoolbox}
\usepackage{siunitx}
\usepackage{color}
\usepackage{xcolor}
\usepackage{soul}

\usepackage{csquotes}
\usepackage{nicefrac}
\usepackage{microtype}
\usepackage{url}

\usepackage[hidelinks]{hyperref}
\usepackage{cleveref}

\usepackage{subcaption}
\usepackage{makecell}
\setcitestyle{numbers,square}

%% file: settings/macros.tex
\newrobustcmd{\ccell}{{\cellcolor{gray!40}}}

\definecolor{tyrianpurple}{HTML}{66023C}

\DeclarePairedDelimiterX{\inp}[2]{\langle}{\rangle}{#1, #2}

\newtheorem{theorem}{Theorem}

\newtheorem{lemma}{Lemma}

\usepackage[table]{xcolor} % Required for cell coloring

\definecolor{hlgreen}{RGB}{210, 235, 215}
\newcommand{\hlg}[1]{\cellcolor{hlgreen}#1}

%% file: sections/00-abstract.tex
\begin{abstract}
Modern online platforms commonly rank ads and organic content separately before blending them into a feed displayed to the user, overlooking externalities: an item's click-through rate depends on its surrounding content, not only on its own position. Recent learning-based feed generation mechanisms model some of these cross-type interactions to globally optimize for the whole feed's welfare. However, these approaches either fix the ordering of organic content, or lack exact strategyproofness guarantees for bidders. To combat these shortfalls, we introduce the \textit{Maximal-in-Range Transformer (MIRT)} mechanism class, which uses a transformer to generate a range of candidate feeds that jointly order ads and organic content, and selects the welfare-maximizing feed in the range. However, there is a tension: strategyproofness requires the generated range to be bid-independent, even though a candidate feed's welfare depends linearly on the bids. Our key technical contribution is a reinforcement learning approach that incorporates both candidate generation and bid-aware selection into training, enabling a \textit{bid-independent} transformer to learn to generate high-welfare ranges by accounting for both individual feed quality and the collective quality of the range. Additionally, we bound the pseudo-dimension of the MIRT class under hard attention, showing that near-optimal expected welfare is learnable with sample complexity polynomial in the transformer size and only logarithmic in the range size. Empirically, MIRT outperforms the previous non-strategyproof state-of-the-art feed models while remaining exactly strategyproof. Our results show that transformer-based auctions can deliver externality-aware whole-feed optimization without sacrificing exact incentive compatibility, removing a major obstacle to their practical deployment.
\end{abstract}

%% file: sections/10-introduction.tex
\section{Introduction}

In modern feed-based ranking systems, organic content and sponsored ads are commonly ranked in independent pipelines and then blended into a single feed \citep{carrion2023blending, li2024deep, nga2025}. This architecture is tractable because it treats a post or ad's click-through-rate (CTR) as primarily a function of its own identity and position relative to other items of the same type. In practice, however, user engagement also depends on the surrounding content. For instance, a hiking post from a close friend may increase the \textit{click-through rate (CTR)} of a nearby outdoor-equipment ad, while several similar ads may instead create fatigue and reduce CTR. This effect is referred to as \textit{permutation externalities}: the CTR of an ad slot or organic post depends not only on its position but also on the surrounding content. When these externalities are significant, this tractability may result in a toll on total welfare. Motivated by this, we focus on the problem of training transformer models for jointly ranking organic and ad content, taking into account permutation externalities.

However, in the advertising context, the goal is not merely to choose a good ranking but also to run an auction. Joint ranking makes it more difficult to enforce \textit{incentive compatibility} among advertisers because an advertiser may benefit from strategic misreporting to influence the content surrounding its ad. Reusing the aforementioned example of an outdoor equipment advertisement: at a low bid, a joint feed aggregator might place a hiking post immediately before the ad, generating a high CTR for the advertiser; after the advertiser raises its bid, the aggregator might promote the ad above and far from the hiking post, so the CTR of the ad decreases because it loses the favorable surrounding context despite receiving a higher slot.  

More formally, a necessary condition for \textit{dominant-strategy incentive compatibility (DSIC)}, also known as \textit{strategyproofness}, is that realized advertiser benefits (e.g., CTR times click benefit) is nondecreasing in the advertiser’s bid, but an unconstrained joint ranker need not satisfy this condition. For many mechanism design problems, the techniques of differentiable economics~\cite{dutting2024optimal} have been used to design state-of-the-art strategyproof mechanisms by training neural networks. For feed generation in particular, recent methods such as CGA \citep{cga2024} and NGA \citep{nga2025} build on this approach. They train transformer models to choose rankings that take into account permutation externalities, and appropriate transfers that minimize incentive compatibility violations. However, these methods can only approximate strategyproofness. In fact, the strategyproofness gap is substantial in practice: in our experiments, we find possible gains from misreporting of more than 10\% of advertiser utility (\Cref{sec:experiments}). To the best of our knowledge, no prior work jointly optimizes the entire ad-organic ordering under permutation externalities while guaranteeing exact DSIC.

To directly address incentive issues, we introduce \textit{Maximal-in-Range Transformers (MIRT)}, a framework that combines the expressive power of transformer-based ranking for welfare-maximization under permutation externalities, with a maximal-in-range~\citep{Nisan07Computationally} allocation rule that ensures DSIC. In particular, MIRT uses a transformer to generate a set of candidate ad-organic joint rankings known as the ``range'', conditioned on ad and organic content, user features, and other constraints, but independent of advertiser bids. At serving time, a maximal-in-range mechanism observes the bids, selects the candidate among the range that maximizes reported welfare, and charges VCG transfers over the range. Using RL, we train the transformer to generate a small \textit{set} of feeds that capture contextual interactions among ads and organic content, accounting for both individual feed quality and the collective quality of the range.

\paragraph{Our Contributions.}
In \Cref{sec:model_and_preliminaries}, we present a general model of joint ad-organic optimization. Prior work either considers permutation externalities only among ads, or models ad-organic interactions while taking the organic ordering as fixed input. In \Cref{sec:methods} we introduce the MIRT mechanism class for exactly strategyproof ad-organic joint optimization under permutation externalities. 
MIRT uses a maximal-in-range approach over a transformer-generated, bid-independent range to ensure strategyproofness. The transformer is trained via RL to generate high-welfare ranges by accounting for both individual feed quality and the collective quality of the range.
To our knowledge, MIRT is the first transformer-based approach to feed optimization that provides a strategyproofness guarantee. 
In \Cref{sec:sample_complexity}, we provide pseudo-dimension bounds for the welfare class induced by MIRT, and show that the sample complexity of learning a MIRT mechanism with optimal expected welfare within the mechanism class is polynomial in the transformer size and only logarithmic in the range size. 
Finally, in \Cref{sec:experiments} we evaluate MIRT in synthetic environments, comparing against the previous state-of-the-art approaches on welfare, total transfer, and DSIC violations, 
and studying the effects of range size and externality strength. 
MIRT attains better welfare than baselines across all tested environments, while also being DSIC by construction. 
By comparison, the previous SOTA approaches have expected regret under truthful bidding as high as 10\% of advertiser utility.
%total transfer was --- here.
\subsection{Related Work} \label{sec:related_work}

\paragraph{Differentiable economics.} A powerful approach to automated mechanism design, sometimes called \textit{differentiable economics}, uses deep neural networks to represent and optimize economic mechanisms. The seminal work of \citet{dutting2024optimal} introduces the RegretNet technique, in which an auction is directly represented as a neural network, which is trained to optimize total transfers while minimizing regret violations. There is a wide literature building on this basic idea by improving network architectures and applying the technique to new mechanism design settings~\citep{pmlr-v162-duan22a,Curry22Learning}. The core limitation of RegretNet is that the guarantee of strategyproofness is only approximate and may be hard to measure. This has motivated the design of specialized network architectures which can train high-performing auctions while ensuring exact strategyproofness by construction, for example by restricting to the class of affine maximizer auctions and learning continuously-parameterized allocation menus~\citep{Curry22Differentiable,Duan23Scalablea}, by using a more computationally heavy but more general menu-based approach~\citep{Wang24GemNet}, or by using reinforcement learning to train a sequential mechanism~\citep{ravindranath2024deepreinforcementlearningsequential,brero2021reinforcement}. Several recent works are of particular interest in the context of MIRT. 
Prior works for feed generation~\cite{nga2025,cga2024} use a RegretNet-inspired training approach (further discussion of these and other works is in \Cref{app:neuralfeed}), exhibiting large DSIC violations. 
\citet{Duan23Scalablea} uses transformer networks to generate the parameters of an affine maximizer (a generalization of MIR mechanisms), but their mechanism design problem is classical additive auction design rather than our feed auctions, their architecture requires a hardcoded range size, and their training techniques search over a continuous space of lottery allocations.

\paragraph{Machine-Learning Powered Preference Elicitation}
While differentiable economics uses a learning algorithm to derive a performant mechanism, a related branch of research embeds the learning algorithm within the final auction mechanism itself, using it to aid preference elicitation in combinatorial settings where the agents' input space grows exponentially. The foundational works of \citet{blum2004preference} and \citet{lahaie2004applying} were the first to frame preference elicitation in CAs as a learning problem. More recently, \citet{brero2018combinatorial, brero2021workingpaper} and \citet{weissteiner2020deep, weissteiner2022fourier, weissteiner2022monotone, weissteiner2023bayesian} introduced ML-powered ICAs for auctions in which agents report cardinal bids, and \citet{Soumalias_2024} extended this approach to auctions driven by demand queries. \citet{soumalias2025_prices_bids_values} then showed that judiciously combining both query types yields substantially higher efficiency in combinatorial settings. The framework has since been extended to settings without transfers \citep{soumalias2024machinelearningpoweredcourseallocation}, and to mechanisms that incorporate LLMs so that agents can express their preferences in natural language \citep{soumalias2025llmpoweredpe}.

\paragraph{Maximal-in-range mechanisms.} 
The Vickrey-Clarke-Groves (VCG) mechanisms choose the allocation that maximizes welfare while using transfers that ensure strategyproofness~\citep{Vickrey61Counterspeculation,Clarke71Multipart,Groves73Incentives}.
If welfare maximization is only approximate, then the corresponding VCG transfers are in general not strategyproof.
One important exception is the class of maximal-in-range (MIR) mechanisms, which optimize exactly over a restricted range of possible outcomes~\citep{Lavi03Characterization,Nisan07Computationally}. As long as the range is chosen without considering bids, the resulting mechanism with VCG transfers is strategyproof. Mechanism design within this class then amounts to designing the range. %transfer was payment
Some existing theoretical work studies approximately welfare maximizing maximal-in-range mechanisms in single-parameter combinatorial auction settings~\citep{Goel10Single,Huang11BlackBox} and in multi-unit auctions~\citep{DobzinskiMechanisms}, as well as sample complexity of learning closely related classes of mechanisms~\citep{balcan16,BalcanSV18}.

%% file: sections/20-model.tex
\section{Model and Preliminaries} \label{sec:model_and_preliminaries}

A platform ranks ads and organic content in a user's feed. Let $\mathcal A=\{a_1,\ldots,a_M\}$ and $\mathcal O=\{o_1,\ldots,o_{N_O}\}$ denote the sets of candidate ads and organic posts, respectively. The set of all candidate items is $\mathcal I=\mathcal A\cup\mathcal O$, of size $N = M + N_O.$ Each advertiser \(j\in[M]\) submits a single ad \(a_j\) and bid \(b_j\in\mathbb R_+\), interpreted as its benefit conditional on a click, and has some underlying true benefit conditional on a click $v_j \in \mathbb{R}^+$. 
We write $B=(b_1,\ldots,b_M)$ for the bid profile. True types (benefits) are drawn from some distribution with CDF $F_i(v_i)$. %type was value
Each organic post \(o\in\mathcal O\) generates welfare \(h_o\geq 0\) when the user engages with it. 
Let  \(e_i\in\mathbb R^d\) be the embedding of item \(i\in\mathcal I\), and let \(C\in\mathcal C\) collect the features of all items.

A feed with horizon \(T\leq M+N_O\) is a sequence $x=(x_1,\ldots,x_T) \in \mathcal{X}.$
For every displayed item $i$, let \(p_x(i)=x^{-1}(i)\) denote its position in the feed, 
and let \(q_i(x)\) denote the rate at which the relevant engagement event occurs in feed $x$.
This rate may depend both on the item's position and on the other items in the feed. 
In particular, we have $q_j(x)=\rho_j(p) \gamma_j(x,p)$, where $\rho_j(p)$ is the baseline, position-dependent click-through rate of ad $a_j$ in the absence of externalities, and \(\gamma_j(x,p)\) is an externality multiplier determined by the content around the ad.

\paragraph{The platform's joint ranking optimization problem.}

The expected welfare generated by feed \(x\) assuming truthful bidding is
\begin{equation}
\label{eq:welfare}
V(x,B) = \sum_{a_j \in x \cap \mathcal{A}} b_j q_j(x) + \sum_{o \in x \cap \mathcal{O}} h_o q_o(x),
\end{equation}
where $h_o$ is the welfare per engagement of organic $o$. Thus, an ad click contributes to the corresponding advertiser's bid, while engagement with an organic post contributes to the platform's welfare.

\subsection{Modeling externalities}

The multiplier $\gamma_j(x,p)$ models how the organic posts surrounding an ad change its click-through rate. 
Our modeling is based on the intuition that placing an ad among relevant organic posts increases its CTR.
We experiment with the previous slot's similarity as the default externality: 
$
\gamma_j(x,p) = 1 + \eta\, s_{j,o_{p-1}}
\, (\eta = 0.5)
$ where $o_{p-1}$ is the item in the slot above $p$, and $s_{j,o_{p-1}}$ is the similarity between that post and the ad.
We explore various other models listed in Appendix \ref{app:externalityfunctions}. 

\subsection{Bidder Models and Transfers}
\label{sec:mechdesignbasics}

\paragraph{Bidder model and strategyproofness} In addition to choosing a feed, we allow charging a transfer $t_i(B)$ to each advertiser based on the bids, and implicitly the chosen feed $x$. A bidder with true benefit $v_i$ may misreport some $b_i.$ 
% We assume that bidders are quasilinear, that is, 
The utility of advertiser $i$ is: $u_i(b_i, v_i, b_{-i}) = v_i q_i(x(B)) - t_i(B).$ We can now define strategyproofness or dominant strategy incentive compatibility: $u_i(v_i, v_i, b_{-i}) \geq u_i(b_i', v_i, b_{-i}),$ for all bidders $i$, true benefits $v_i$, opponent bid profiles $b_{-i}$, and possible misreports $b_i'.$
% benefit was value
\paragraph{The VCG mechanism}
The platform's goal is to select the welfare-maximizing feed $x^\star \in \arg\max_{x\in \mathcal{X}}V(x,B)$. The Vickrey-Clarke-Groves mechanism is a strategyproof mechanism with $x^*$ as its allocation rule~\citep{Vickrey61Counterspeculation,Clarke71Multipart,Groves73Incentives}, which charges bidders the externality they cause to everyone else by being present in the auction: $t_i(B) = \max_{x \in \mathcal{X}} V(x, B_{-i}) - (V(x^*, B) - b_iq_i(x^*)).$ As discussed below, exactly computing $x^*$ as required for VCG is intractable, so we use MIR mechanisms instead, which have the same allocation and transfer rule but with feed choices restricted to $x \in R.$

\begin{figure}[t]
    \centering
    \caption{Schematic diagram of the MIRT architecture.}
    \includegraphics[width=\linewidth]{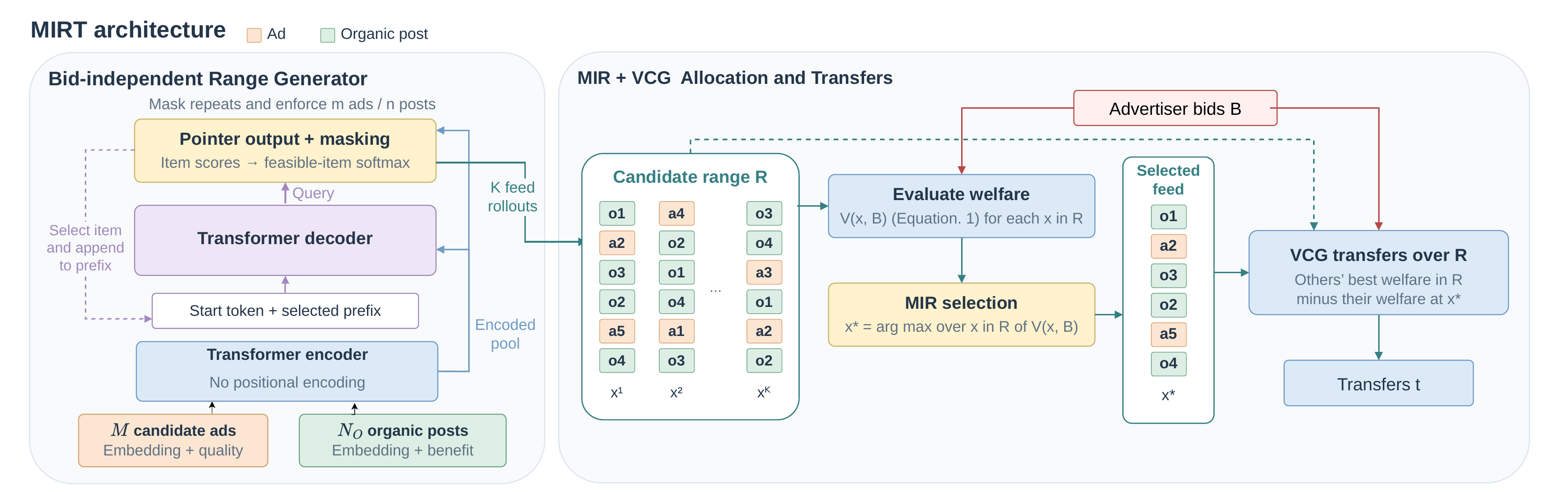}
    \label{fig:alg}
    \vspace{-3mm}
\end{figure}

%% file: sections/30-auctions.tex
\section{Learning Maximal-in-Range Mechanisms}
\label{sec:methods}

The ideal solution to the feed-generation problem would be to run a VCG advertising auction. A VCG mechanism selects a welfare-maximizing allocation out of all possible allocations and calculates transfers for the corresponding allocation. In our advertisement auction setting, the allocation is an ordered selection of $m$ ads and $n$ organic posts out of a pool of $M$ ads and $N_O$ posts, which we call a feed. The number of possible feeds is $\binom{M}{m}\binom{N_O}{n}(m+n)!$, making exhaustive welfare evaluation intractable. Instead, we use a \textit{maximal-in-range (MIR)} mechanism that selects the welfare-maximizing allocation from a smaller set of candidates (range), preserving incentive compatibility as long as the range is generated independently of advertiser bids~\cite{Nisan07Computationally}.

An MIR mechanism is defined by a range $R$ of candidate feeds. This range may depend on the context $C$, including the true ad qualities, organic benefits, item embeddings, and user features, but must be independent of advertiser bids to stay DSIC. Once the range is fixed and bids are observed, the mechanism performs well as long as one element has high performance. Therefore, a good MIR mechanism must predict high quality feeds from the context, and choose a combination of such feeds that covers as well as possible the various bid profiles which may arrive.

Our range generation model with parameters $\theta$ can produce a context dependent range $R_\theta(C)$. Given this range, the mechanism evaluates candidate welfare using the submitted bids, selects the allocation \(a_\theta(B,C)\), and computes VCG transfers over the same range:
$
a_\theta(B,C)\in\arg\max_{x\in R_\theta(C)}V(x,B).
$
Motivated by the sequential nature of feed construction and the prior works' models, we propose an RL-trained encoder-decoder architecture to generate a feed by selecting ads and organic posts autoregressively from a candidate pool of items. The model takes item embeddings, organic post benefits, and ad qualities as inputs and generates one feed per rollout. A collection of $K$ distinct feeds creates the candidate range $R_\theta(C)$ for the MIR mechanism.

\subsection{Neural architecture} 

We use an encoder-decoder transformer \citep{vaswani2017attention} with a pointer output \citep{vinyals2015pointer, kool2019attention} to construct feeds from bid-independent item features (1.38M parameters). 
The encoder has no positional encoding, therefore it is permutation-equivariant \citep{zaheer2017deep} and reordering the candidate pool does not reorder its representations. The pool is encoded once per range and shared across all $K$ rollouts \citep{kwon2020pomo}. The decoder then builds a feed autoregressively, at each step applying causal self-attention over the items chosen so far, cross-attending to the encoded pool, and scoring every remaining candidate with a clipped scaled dot-product pointer \citep{bello2017neural, kool2019attention}. Decoder inputs carry a learned positional encoding \citep{gehring2017convolutional}.

% novelties in our design...
Our design has the following considerations: first, a mask at every decoder step blocks items already placed or the items that would violate the feed's composition, so every rollout is a valid feed by construction rather than by a penalty term. Second, the VCG mechanism needs a diverse menu of feeds both to produce higher welfare and transfer assignments~\citep{guzman2012multiple, grinsztajn2023winner}. In our design, each rollout is conditioned on a per-rollout embedding~\citep{carion2020end, chalumeau2023combinatorial} so that the $K$ rollouts of one context produce $K$ different feeds instead of $K$ copies of similar feeds.

\subsection{Training algorithm}

We train the range generator using on-policy policy-gradient reinforcement learning \citep{williams1992simple}. The policy $\pi_\theta$ constructs each feed item by item per rollout, conditioned on the context $C$, repeated $K$ times to generate the range. To score a range, the environment draws $D=8$ bid profiles from the prior $F$ and takes the welfare of the best feed in the range under each, so the reward is the expected best-of-$K$ welfare, $r(R;C)=\mathbb{E}_{B\sim F}\left[\max_{x\in R}V(x,B)\right]$.

We follow POMO~\citep{kwon2020pomo}, which builds on REINFORCE policy gradient with an entropy bonus and a baseline shared across rollouts of the same instance rather than a learned critic. For each context, we generate $G=8$ ranges for the baseline, force their rollouts to begin with different items so they explore different feeds, and define each range's advantage as its reward minus the mean reward of the other $G-1$ ranges \citep{kool2019buy}, scoring all $G$ on the same $D$ bid profiles. Advantages are standardized and centered across batches. We additionally encourage exploration through forced initial selections and added rollout-diversity vectors. During training, we force distinct initial items for up to $\min(K,N)$ rollouts following \citet{kwon2020pomo} and let the remaining rollouts sample their first item from the policy, including those sampled initial choices. For all training details, including network architecture and optimizer details, see \Cref{app:mirttrainingapp}.

\section{Sample Complexity of MIRT} \label{sec:sample_complexity}

We analyze the statistical generalization of MIRT under hard attention when the range is dynamically generated via beam search. Let $\mathcal{C}$ be the space of contexts and $\mathcal{B}$ the space of bid profiles. An auction instance is defined as a tuple $S = (C, B)$, where $C \in \mathcal{C}$ contains the features of $N$ candidate content items, and $B \in \mathcal{B}$ contains the corresponding bids. 
All detailed proofs can be found in \Cref{app:fullproofs}.

As in \citet{yang2026tightcomplexity}, let $s_\theta$ denote the score computed by an autoregressive transformer without layer normalization, with piecewise-linear activation functions and hard attention, weights $\theta \in \mathbb{R}^W$ and \textit{computational depth} $L$, defined as the total number of sequential attention and feed-forward layers. At each generation step $t \in \{1, \dots, T\}$, the transformer scores a candidate piece of content $a$ given the candidate pool $C$ and the items placed so far. Writing $C(x_{<t})$ for the features of the items in the first $t-1$ slots, 
we denote this score as $s_\theta(a \mid C, x_{<t}).$

MIRT generates a discrete range of candidate feeds $R_\theta(C)$ of cardinality at most $K$ using deterministic beam search of width $K$ over $T$ content slots, with a fixed tie-breaking rule.  
The welfare function evaluated on instance $S$ is
$w(\theta; S) = \max_{x \in R_\theta(C)} V(x, B), $
where $V(x, B)$ is the welfare of feed $x$ under bids $B$.
We define the hypothesis class of welfare functions parameterized by the network as $\mathcal{W} = \{ w(\theta; \cdot) \mid \theta \in \mathbb{R}^W \}$.

First, we state our main theoretical result, \Cref{thm:pdim}:

\begin{theorem}[Pseudo-dimension of MIRT] \label{thm:pdim}
Let $\mathcal{W}$ be the welfare function class induced by the transformer class defined above, with computational depth $L$ and  $W \ge L$ parameters, where beam search of width $K$ generates ranges of feeds of length $T$ from a set of $N$ candidate pieces of content. 
Then,  
\begin{equation}
   \operatorname{Pdim}(\mathcal{W}) 
   = {O} \big( W  L  T \log( TW ) \log (T K N) \big). 
\end{equation}
\end{theorem}

\Cref{thm:pdim} enables us to bound the sample complexity of learning an almost-optimal member of the MIRT class using the standard pseudo-dimension uniform-convergence guarantee of \citet{DBLP:journals/jcss/LiLS01}:

\begin{theorem}[Sample complexity of MIRT]
\label{cor:sample-complexity}
Let $V_{\max}$ be the maximum possible welfare. 
Then, for any distribution over auction instances $\mathcal{D}$, the sample complexity of learning with probability at least 
$1 - \delta$ a parametrization of $\mathcal{W}$ with expected welfare within $\epsilon$ of the optimal achievable by the class is: 

\begin{equation}
S = O\left(
\frac{V_{\max}^2}{\epsilon^2}
\left (
WLT
\log(TW)
\log(TKN)
\log\left(\frac{V_{\max}}{\epsilon}\right)
+
\log\left(\frac{1}{\delta}\right)
\right )
\right).
\end{equation}
\end{theorem}

\paragraph{Proof Idea.} To bound $\text{Pdim}(\mathcal{W})$, we must bound the maximum number of auction instances $m$ that $\mathcal{W}$  can shatter. The proof proceeds in five main steps:

\begin{enumerate}
\item \textbf{Reducing from bounding welfare cardinality to bounding range cardinality (\Cref{lemma:bottleneck}):}
The Transformer parameters $\theta$ affect welfare only through the generated range $R_\theta(C)$. In particular, for any fixed auction instance, two parameter vectors that generate the same range yield exactly the same welfare. 
Thus, the number of distinct welfare-threshold labelings that \(\mathcal W\) can generate on any fixed collection of auction instances is upper bounded by the number of distinct range tuples it can generate.
\item \textbf{Characterizing range generation through pairwise comparisons (\Cref{lemma:comparisons}):}
Beam search can be viewed as a sequence of branching decisions. At each autoregressive step, the next beam state is determined by pairwise comparisons between the Transformer scores of at most $KN$ candidate continuations. Hence, conditional on the current beam state, at most
$q=\binom{KN}{2}=O(K^2N^2)$
pairwise comparisons determine the next beam.
\item \textbf{Bounding the number of realizable range tuples by the VC dimension of the Transformer binary comparison class (\Cref{lemma:growth}):}
Let $\mathcal{H}_{\mathrm{cmp}}$ denote the binary hypothesis class induced by these pairwise Transformer-score comparisons, with VC dimension $d_{\mathcal{H}_{\mathrm{cmp}}}$. 

Conditional on any realizable beam-state history, the comparison inputs at the next generation step are fixed. Thus, we can apply the growth function of $\mathcal{H}_{\mathrm{cmp}}$ sequentially across the $T$ generation steps. Sauer's lemma \citep{SAUER1972145} then bounds the number of distinct range tuples that can be generated on $m$ fixed contexts as a function of the VC dimension of $\mathcal{H}_{\mathrm{cmp}}$.

\item \textbf{Bounding the VC dimension of the binary comparison class (\Cref{thm:pairwise-transformer-vc}):}
We adapt the parameter-space partitioning argument of \citet{yang2026tightcomplexity}. 
For any fixed collection of inputs, \citet{yang2026tightcomplexity} showed that the parameter space can be partitioned into regions on which all Transformer branching decisions are fixed and the resulting Transformer scores are polynomial in the parameters of degree $O(L)$ (\Cref{lem:pairwise-partition}). 
Thus, within each such region, every pairwise score comparison is the sign of a difference of two degree-$O(L)$ polynomials, which is itself a polynomial of degree $O(L)$. Warren's theorem \citep{Warren1968LowerBF} bounds the number of realizable comparison patterns (\Cref{lem:pairwise-local-patterns}), yielding
$d_{\mathcal H_{\mathrm{cmp}}} = O\!\left( WL\log(TW) \right)$.

\item \textbf{From range growth to pseudo-dimension (\Cref{thm:pdim}):}
Pseudo-shattering $m$ auction instances require at least $2^m$ distinct range tuples. Comparing this requirement with the range-growth bound gives
$\operatorname{Pdim}(\mathcal W)
=
O\!\left(
d_{\mathcal H_{\mathrm{cmp}}}
T\log(TKN)
\right).$ Substituting the comparison-class VC-dimension bound yields
$\operatorname{Pdim}(\mathcal W)
=
O\!\left(
WLT\log(TW)\log(TKN)
\right).$
\end{enumerate}

%% file: sections/40-experiments.tex
\section{Experiments}
\label{sec:experiments}

In this section we present numerical experiments that establish the effectiveness of our proposed MIRT framework compared to baseline mechanisms, including previous methods~\citep{cga2024,nga2025}.

\subsection{Default Synthetic Environment}
\label{subsec:environment}

To sample our organic posts and ad embeddings, we consider 50 topic centroids and we sample the ads and posts around those. The default environment consists of 80 sampled ads and 300 sampled organic posts per centroid, therefore 15,000 posts and 4,000 ads. This environment contains 50,000 users, and a pool of $M=10$ ads and $N_O=18$ posts per user. Therefore, the problem size is \((M,N_O,T)=(10,18,24)\) with a fixed feed length of 24 and a fixed six ads per feed. In this setting, ad qualities follow a lognormal distribution with the cosine similarity between user-item embeddings as the log-mean and a log-standard deviation of $0.5$. Organic utilities follow a Weibull distribution with a mean of exponentiated user-item cosine similarity and a shape of $0.6$. (Appendix Figure \ref{fig:welfare_distributions}) The default bid prior is exponential with rate 1 and no shift: (mean 1, standard deviation 1). Details on synthesizing the environment can be found in Appendix \ref{sec:env_setup}.
Additionally, the default environment reduces the impression weight by $0.02$ with each position down the feed and also reduces the conversion weight for ads by $0.01$. Finally, the externality function for this environment takes into account the similarity between the ad and its previous posts, also introduced as Position-modulated externality in Appendix~\ref{app:externalityfunctions} with $\eta_{\mathrm{above}}=0.5$ and $\eta_{\mathrm{below}}=0$.

\subsection{Experimental Setup}

We designed our experiments to explore our method's robustness with respect to different problem sizes, bid distributions, and externality functions.

\paragraph{Environment variants.} Starting from the default environment (\((M,N_O,T)=(10,18,24)\)), we vary one component at a time. We experiment with Exponential$(1)$ (mean $1$), Uniform$(0.001, 4)$ (mean $\approx 2$), LogNormal$(\mu{=}0, \sigma{=}0.5)$ (mean $e^{0.125} \approx 1.13$), and Pareto$(\alpha{=}1.5, x_m{=}1)$ (mean $3$, infinite variance) distributions. 
We train models for each externality model, and various problem sizes by considering different numbers of candidate ads, organic posts, and feed lengths.
For the additive, activation, attention, multiplicative, and topic externalities, we use a neighborhood radius of $w=4$, covering organic posts up to four positions before or after the ad, and an externality strength of $\eta=0.5$. 
For more details on the externality functions, please see Appendix \Cref{tab:externality_params}.

\paragraph{MIRT variants.}
We evaluate MIRT using ancestral sampling with range sizes $K\in\{32,64,128,256,512\}$, (denoted MIRT $K=.$). Every trained model, although trained on $K=32$, is capable of generating any range size. MIRT range generation is independent of bids, so it is exactly DSIC. However, varying the bid distributions, problem sizes, and externality modelings will require separate model training.

\paragraph{Training details.} Among the 50,000 dataset samples, we used 40,000 for training, 7,500 for testing and 2,500 for validation. The training hyperparameters follow the description in \Cref{sec:methods}, with $K=32$.  The ancestral decoding temperature defaults to 1 and, where mentioned, is tuned on a validation set (see \Cref{fig:temperatureanalysis} for the effect of temperature).

\subsection{Baselines and evaluation methods}

\paragraph{MIR baselines.} 

As our main baseline, we consider Random+MIR, which builds a bid-independent range of $K=32$ feeds, each drawn uniformly at random from the feasible feeds. MIR then selects the feed with the highest welfare at the reported bids.

We also consider a bid-blind greedy method to generate the range (Greedy+MIR). Greedy+MIR scores each item by its bid-independent welfare, meaning the welfare generated by an ad's quality and a post's user benefit, and not the bids. It then builds a feed by keeping the highest-scoring feasible set of items and placing the best ones in the most prominent slots. The first feed follows this ranking exactly. The remaining $K-1$ feeds perturb the scores with random noise, so the range holds $K$ distinct high-quality feeds.

\paragraph{Comparisons to prior work.}

We re-implement modified versions of CGA and NGA~\citep{cga2024,nga2025}, adapted to our setting (see Appendix \ref{app:cgangatraining} for details). These methods have terms in the loss function to encourage incentive compatibility during training, but do not guarantee exact DSIC. CGA and NGA regret hyperparameters ($\rho$, the strength of regret penalty) were tuned so they produce lower regret values to become closer to our zero regret setting. We found that additional regret minimization lowered their performance. It is possible that different hyperparameters could unlock improved performance with low regret; on the other hand, it is also possible that the empirical method for estimating regret fails to detect some misreports. We see these ambiguities as a fundamental limitation of approximately-strategyproof methods for neural mechanism design.

\paragraph{Bid-aware baselines.}
To construct a strong welfare benchmark for this setting, we used bid-aware beam search. This method does not generate a range for MIR, but by calculating the welfare produced by every item (including actual bids in the ad welfare), it produces a single high welfare feed with beam-search algorithms with 4 and 8 width. This method is a strong welfare benchmark, but because it cannot be combined with an appropriate transfer rule, it cannot be interpreted as a mechanism.

\paragraph{Performance evaluation.}

We compare MIRT variants with Random+MIR, Greedy+MIR, CGA, NGA, and bid-aware greedy and beam-search benchmarks. 
We evaluate MIRT using ancestral sampling. We report organic welfare, advertiser welfare, (total) social welfare, total transfer, and empirical regret, which we denote $\Psi$. Regret is the additional utility an advertiser could have gained by misreporting their bid, measured after the bids are known~\cite{dutting2024optimal}. When comparing all other methods, we report either absolute metrics or lift over Random+MIR.

\begin{figure}[t]
    \centering
    \caption{\textbf{Performance in the reference environment.} Box-plot with medians/spread for test contexts for one trained model. The first two panels show per-context percentage lift over Random+MIR in total welfare and total transfer. The final panel shows ex-post regret as a percentage. MIRT used sampling temperature 2.4, and CGA/NGA used $\rho=1000.$ MIRT achieves higher total welfare than CGA and NGA, and has zero regret by construction.}
    \includegraphics[width=0.98\linewidth]{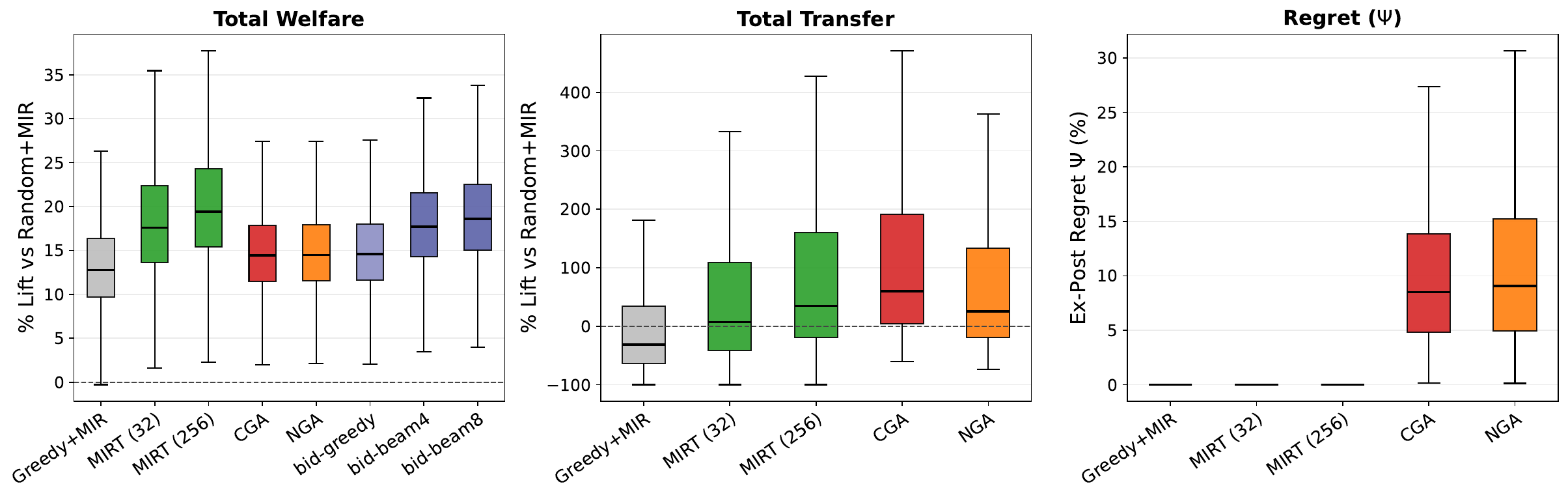}
    \label{fig:toplines}
\end{figure}

% \mjc{i removed "Results" as a section"}

\subsection{Main results for reference configuration}
\label{subsec:mainresults}
We first compare all methods under the reference configuration described above. We evaluate social welfare, total transfer, and incentive properties. Results are shown in \Cref{fig:toplines} and Appendix~\Cref{tab:absolute_experiment1} includes more detail. \Cref{app:virtualwelfare} also contains experiments training for the modified objective of virtual welfare (to maximize total transfer). 
For our main objective of total welfare, MIRT outperforms all other mechanisms, including our Greedy+MIR baseline, and crucially, the prior SOTA approaches in the literature, CGA and NGA.
For total transfer, MIRT outperforms the Greedy+MIR baseline, while being marginally below CGA, and on par with NGA. 
Importantly, MIRT achieves this performance while being exactly strategyproof by construction. By comparison, CGA and NGA exhibit a very high amount of regret.

\subsection{Varying externality functions}
\label{sec:externality-experiments}

\begin{table}[htbp]
\centering
\caption{Mean ($\pm$ approximate 95\% CI) absolute total welfare across held-out contexts under each externality function with default problem size and bid distribution. Highlighted cells are the highest-performing mechanisms. Corresponding total transfer is in Appendix~\Cref{tab:revenuefromorigtable}. This table was generated with 40 epochs of training for each MIRT model, and sampling temperature 1 (except for the Prev. Organic externality, which used 2.4). CGA/NGA used $\rho=1.$}
\label{tab:externality}
\resizebox{\linewidth}{!}{
\begin{tabular}{l *{7}{c}}
\toprule
Externality & Prev. Organic & Additive & Activation & Attention & Multiplicative & Topic & Position \\
\midrule
Random+MIR & $27.18 \pm 0.214$ & $27.34 \pm 0.215$ & $29.67 \pm 0.227$ & $28.50 \pm 0.220$ & $27.49 \pm 0.216$ & $27.18 \pm 0.214$ & $27.33 \pm 0.214$ \\
Greedy+MIR & $30.71 \pm 0.239$ & $30.83 \pm 0.240$ & $33.24 \pm 0.255$ & $31.93 \pm 0.245$ & $30.91 \pm 0.242$ & $30.63 \pm 0.237$ & $30.93 \pm 0.241$ \\
MIRT $(K = 32)$ & $32.11 \pm 0.245$ & $31.80 \pm 0.243$ & $33.92 \pm 0.257$ & $32.41 \pm 0.247$ & $33.00 \pm 0.254$ & $31.87 \pm 0.247$ & $32.86 \pm 0.250$ \\
MIRT $(K = 256)$ & \hlg{$32.50 \pm 0.246$} & \hlg{$32.20 \pm 0.245$} & \hlg{$34.47 \pm 0.258$} & \hlg{$32.86 \pm 0.248$} & \hlg{$33.49 \pm 0.256$} & \hlg{$32.34 \pm 0.249$} & \hlg{$33.30 \pm 0.252$} \\
CGA & $31.08 \pm 0.238$ & $31.37 \pm 0.241$ & $33.55 \pm 0.254$ & $32.56 \pm 0.247$ & $31.35 \pm 0.241$ & $31.20 \pm 0.239$ & $31.15 \pm 0.239$ \\
NGA & $31.09 \pm 0.238$ & $31.33 \pm 0.240$ & $33.51 \pm 0.254$ & $32.59 \pm 0.247$ & $31.37 \pm 0.241$ & $31.23 \pm 0.239$ & $31.16 \pm 0.239$ \\
\midrule
bid-greedy & $31.12 \pm 0.238$ & $31.36 \pm 0.241$ & $33.54 \pm 0.254$ & $32.61 \pm 0.247$ & $31.39 \pm 0.241$ & $31.23 \pm 0.239$ & $31.19 \pm 0.239$ \\
bid-beam4 & $31.96 \pm 0.243$ & $32.29 \pm 0.244$ & $35.52 \pm 0.262$ & $33.47 \pm 0.251$ & $32.54 \pm 0.249$ & $32.20 \pm 0.247$ & $32.44 \pm 0.246$ \\
bid-beam8 & $32.20 \pm 0.244$ & $32.42 \pm 0.245$ & $35.58 \pm 0.262$ & $33.53 \pm 0.251$ & $32.78 \pm 0.250$ & $32.35 \pm 0.247$ & $32.77 \pm 0.248$ \\
\bottomrule
\end{tabular}
}
\end{table}

We examine whether MIRT's welfare performance is robust to different forms of permutation externalities, as described in Appendix \ref{app:externalityfunctions}.
We vary the externality function while holding all other experimental parameters fixed. Results are in \Cref{tab:externality}. 
For our main objective of welfare, MIRT outperforms all other mechanisms, for all externality functions. 
For total transfer (Appendix~\Cref{tab:revenuefromorigtable}), the best-performing mechanism depends on the externality function. However, again, MIRT has zero regret by construction, while the other mechanisms exhibit high amounts of regret.

\subsection{Varying bid distributions}
\label{sec:bid-distribution-experiments}

In \Cref{tab:bid_dist}, we evaluate the robustness of MIRT's performance to different advertiser populations by varying the distribution of advertiser bids. 
Again, for our central objective of Social Welfare, MIRT outperforms all other mechanisms by a wide margin, even at $K = 32$. 
For uniform, lognormal, and Pareto distributions, it even outperforms all other mechanisms in terms of total transfer.

\begin{table}[t]
\centering
\caption{Mean ($\pm$ CI across held-out contexts) absolute total welfare and total transfer across held-out contexts under each bid prior and default problem size and externality. Highlighted cells are the highest-performing mechanisms. Results in this table were generated with 40 epochs of training for each MIRT model, with temperature $2$. CGA/NGA have $\rho = 1000$.}
\resizebox{\linewidth}{!}{%
\begin{tabular}{lcccccccc}
\toprule
& \multicolumn{4}{c}{Social Welfare} & \multicolumn{4}{c}{Total Transfer} \\
\cmidrule(lr){2-5} \cmidrule(lr){6-9}
Distribution & Exponential & Uniform & Lognormal & Pareto & Exponential & Uniform & Lognormal & Pareto \\
\midrule
Random+MIR & $27.15 \pm 0.214$ & $34.27 \pm 0.222$ & $27.33 \pm 0.205$ & $44.86 \pm 0.761$ & $2.05 \pm 0.039$ & $4.83 \pm 0.073$ & $2.47 \pm 0.040$ & $5.60 \pm 0.093$ \\
Greedy+MIR & $30.71 \pm 0.240$ & $39.83 \pm 0.257$ & $31.38 \pm 0.227$ & $50.59 \pm 0.781$ & $1.44 \pm 0.023$ & $3.29 \pm 0.048$ & $1.93 \pm 0.026$ & $3.72 \pm 0.058$ \\
MIRT $(K = 32)$ & $32.06 \pm 0.245$ & $42.67 \pm 0.268$ & $32.69 \pm 0.230$ & $53.49 \pm 0.882$ & $2.23 \pm 0.035$ & $5.55 \pm 0.075$ & $2.76 \pm 0.032$ & $5.05 \pm 0.073$ \\
MIRT $(K = 256)$ & \hlg{$32.52 \pm 0.246$} & \hlg{$43.32 \pm 0.269$} & \hlg{$32.95 \pm 0.231$} & \hlg{$55.02 \pm 0.916$} & $2.76 \pm 0.037$ & \hlg{$7.14 \pm 0.081$} & \hlg{$3.75 \pm 0.036$} & \hlg{$6.64 \pm 0.076$} \\
CGA & $31.08 \pm 0.238$ & $40.37 \pm 0.252$ & $31.19 \pm 0.226$ & $51.00 \pm 0.767$ & \hlg{$3.26 \pm 0.024$} & $6.35 \pm 0.024$ & $3.29 \pm 0.011$ & $3.25 \pm 0.015$ \\
NGA & $31.09 \pm 0.238$ & $40.31 \pm 0.252$ & $31.20 \pm 0.226$ & $50.98 \pm 0.766$ & $2.50 \pm 0.016$ & $5.97 \pm 0.022$ & $2.90 \pm 0.009$ & $2.22 \pm 0.014$ \\
\midrule
bid-greedy & $31.12 \pm 0.238$ & $40.39 \pm 0.252$ & $31.22 \pm 0.226$ & $51.02 \pm 0.766$ & -- & -- & -- & -- \\
bid-beam4 & $31.96 \pm 0.243$ & $41.50 \pm 0.260$ & $32.06 \pm 0.229$ & $52.73 \pm 0.813$ & -- & -- & -- & -- \\
bid-beam8 & $32.20 \pm 0.244$ & $41.87 \pm 0.262$ & $32.27 \pm 0.230$ & $53.55 \pm 0.839$ & -- & -- & -- & -- \\
\bottomrule
\end{tabular}}
\label{tab:bid_dist}
\end{table}
\subsection{Varying range size}
\label{sec:range-size-experiments}
MIR mechanism welfare improves as more items are added to the range; when the range is large enough to include all feeds, exact VCG is recovered. In \Cref{fig:range_size_scaling}, we show the user welfare, advertiser welfare, total welfare, and total transfer in the reference environment as we vary the range over \(K\in\{32,64,128,256,512\}\), along with latency. MIRT ran with decoding temperature 2. Because MIRT runs in parallel on the GPU, its latency does not increase with the range size. Unsurprisingly, performance in both welfare and total transfer improves as the size of the range increases.
\begin{figure}[h]
    \centering
    \caption{Scaling of performance with the size of the range.}
    \includegraphics[width=0.98\linewidth]{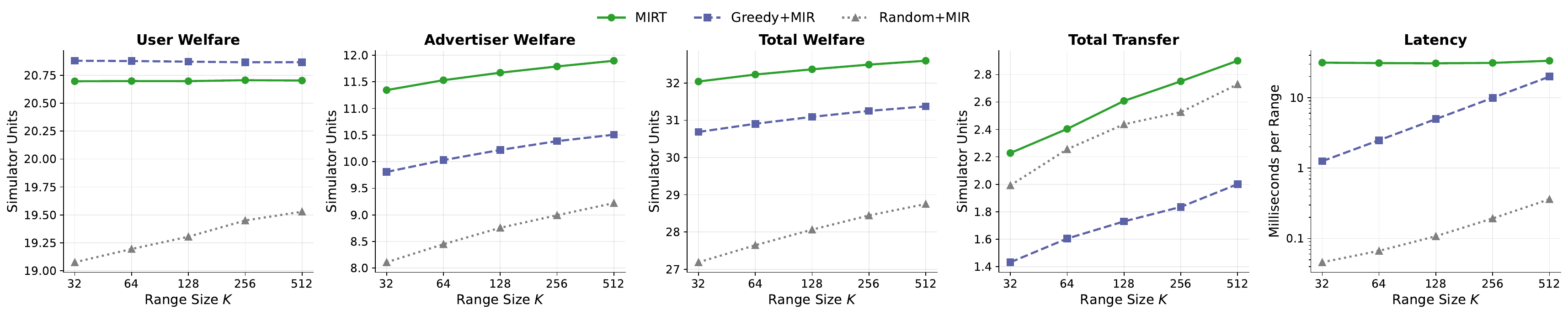}
    \label{fig:range_size_scaling}
    \vspace{-3mm}
\end{figure}

%% file: sections/50-conclusion.tex
\section{Discussion and Conclusion}
\label{sec:conclusion}

We present MIRT, an exactly strategyproof mechanism for whole-feed optimization. MIRT uses an autoregressive transformer to generate the feeds in the range of a maximal-in-range mechanism, optimized via a reinforcement learning signal that takes into account both individual feed and collective range quality. We show that with hard-attention transformers, near-optimal expected welfare is learnable with sample complexity polynomial in the transformer size and only logarithmic in the range width. 
In extensive experiments, we show that MIRT outperforms prior state-of-the-art approaches in welfare across all settings tested, and remains competitive in total transfer. 
At the same time, these prior approaches for joint ranking show high strategyproofness violations in our experiments, while MIRT ensures exact strategyproofness by construction. There are many promising directions for future work, including testing improved RL algorithms, extending our methodology to broader classes of mechanisms such as affine maximizers, and applying the same basic approach to many other mechanism design problems beyond feed optimization.

\clearpage
\newpage
\subsection*{AI use statement}

We used generative AI in several ways in our paper. We had AI tools review our draft and provide suggestions for improved language as well as high-level suggestions about which parts of our argument were more or less convincing.  We used AI as a search tool in a human-driven literature search to surface papers we might not have been aware of, in addition to related work found manually or already known to the authors. We used AI coding assistants extensively for our code, which includes matplotlib code to generate plots.
The proof in \Cref{sec:sample_complexity} is primarily human work, but conversational AI assistance was used to suggest some counting lemmas to tighten the bound, as well as to edit the presentation of the proof for clarity and correctness.

%% file: sections/60-appendix.tex
\appendix
\section{Related work: neural feed allocation mechanisms}
\label{app:neuralfeed}
There are several works using neural networks to learn feed-allocation auction mechanisms, including \citet{genauction2025,egav2_2025}. We highlight two approaches of particular relevance to this work, CGA~\citep{cga2024} and NGA~\citep{nga2025}, both of which consider feed generation in the presence of externalities, and both of which use RL to train their feed generators.
\citet{cga2024} employs an autoregressive decoder that takes context embeddings as input and directly generates ad allocations to maximize expected virtual welfare, taking into account externalities. Their objective is to maximize virtual welfare; the virtual-welfare-maximizing allocation rule has corresponding strategyproof Myerson transfers. However, CGA does not learn the exact optimal allocation rule, and moreover, computing Myerson transfers (via numerical integration) would be infeasibly slow. Therefore, they use a decoupled PaymentNet module to approximate correct transfers, and empirically minimize regret in the fashion of \citet{dutting2024optimal}. The result is an auction that performs well, but for which exact strategyproofness cannot be ensured.
\citet{nga2025} introduces a non-autoregressive feed generator with dual encoders for candidate items and candidate positions, and an evaluator that outputs predicted click-through and conversion rates, along with transfers, which are used to select a final feed.  The networks are also trained using an approach similar to \citet{dutting2024optimal} to encourage incentive compatibility. Our approach also uses transformers trained via RL to generate feeds, but we achieve exact strategyproofness and generate our entire range of feeds simultaneously.

\section{Externality functions}
\label{app:externalityfunctions}

Externality models ($\gamma_j(x,p)$) describe the interaction between ads and surrounding posts. In the following list, we present the different externality modeling functions.

Externalities are defined for a position $p$, in a feed $x$, filled with advertisement $a_j$. The ad $a_j$ has embedding $e_j$, and a surrounding organic post $o$ has embedding $e_o$. We write $s_{jo}=\cos(e_j,e_o)$ for their similarity, and $\kappa_x(o,p)=\lambda^{|\,p_x(o)-p\,|}$ for the decay factor that shrinks as $o$ gets further from $p$. Most of these models have a strength hyperparameter $\eta$ that defines the intensity of the interactions.

    \paragraph{Additive similarity.} This method sums up the nearby organic post embeddings with a decay factor based on distance from position $p$, over the window of slots $p-w$ to $p+w$:
     \[
        \gamma_j(x,p)
        =
        1+\eta\,
        \cos\!\Big(e_j,\ \sum_{|\mathrm{pos}_x(o)-p|\,\leq\, w}\kappa_x(o,p)\,e_o\Big),
        \qquad \eta\geq 0 ,
    \]

    \textbf{Neighborhood activation.}
    The ad receives a fixed boost whenever at least one organic post in slots $p-w$ to $p+w$ is
    sufficiently similar to it, where $s_{jo}=\cos(e_j,e_o)$. 
    \[
    \gamma_j(x,p)
    =
    1+\eta\,
    \mathbf 1
    \left\{
    \max_{|\mathrm{pos}_x(o)-p|\,\leq\, w}s_{jo}
    >
    \tau_{\mathrm{act}}
    \right\}.
    \]

    \textbf{Softmax attention.}
    The ad attends to all nearby organic posts, with more relevant posts
    receiving greater weight. Define
    \[
    \alpha_{jo}(x,p)
    =
    \frac{
        \kappa_x(o,p)\exp\!\left(\langle e_j,e_o\rangle/\tau_{\mathrm{sm}}\right)
    }{
        \displaystyle
        \sum_{|\mathrm{pos}_x(o')-p|\,\leq\, w}
        \kappa_x(o',p)\exp\!\left(\langle e_j,e_{o'}\rangle/\tau_{\mathrm{sm}}\right)
    },
    \qquad
    \tau_{\mathrm{sm}}>0 .
    \]
    where $\langle e_j,e_o\rangle$ is the unnormalized inner product of the two embeddings. The externality multiplier is:
    \[
    \gamma_j(x,p)
    =
    1+\eta
    \sum_{|\mathrm{pos}_x(o)-p|\,\leq\, w}
    \alpha_{jo}(x,p)\,s_{jo}.
    \]

    \textbf{Multiplicative complements and substitutes.}
    Each organic post in slots $p-w$ to $p+w$ independently scales the ad's CTR, so the effects
    compound instead of averaging:
    \[
    \gamma_j(x,p)
    =
    \prod_{|\mathrm{pos}_x(o)-p|\,\leq\, w}
    \left(
    1+\eta\,\kappa_x(o,p)\,s_{jo}
    \right).
    \]
    Positive \(s_{jo}\) represents complementary content, whereas negative \(s_{jo}\) represents
    substitute or competing content.

    \textbf{Topic-category matching.}
    Affinity is measured in a shared space of $K$ latent topics rather than in raw embedding space.
    Let $\phi_{j,k}$ denote ad $j$'s affinity for topic $k$ and $\psi_{o,k}$ organic post $o$'s
    loading on topic $k$, both distributions over the $K$ topics. The surrounding context is the
    weighted mixture of its posts' loadings,
    \[
    \bar\psi_k(x,p)
    =
    \frac{
        \sum_{|\mathrm{pos}_x(o)-p|\,\leq\, w}\kappa_x(o,p)\,\psi_{o,k}
    }{
        \sum_{|\mathrm{pos}_x(o)-p|\,\leq\, w}\kappa_x(o,p)
    },
    \]
    and the ad benefits when its preferred topics are the ones that the mixture is rich in:
    \[
    \gamma_j(x,p)
    =
    1+\eta
    \left(
    K\sum_{k=1}^{K}\phi_{j,k}\,\bar\psi_k(x,p)
    \;-\;1
    \right).
    \]

    \textbf{Position-modulated.}
    
    Writing $o_{p-1}$ and
    $o_{p+1}$ for the items in the slots directly above and below $p$,
    \[
    \gamma_j(x,p)
    =
    1+\eta_{\mathrm{above}}\,s_{j o_{p-1}}
    +\eta_{\mathrm{below}}\,s_{j o_{p+1}},
    \qquad
    \eta_{\mathrm{above}}>\eta_{\mathrm{below}} ,
    \]
    where a term contributes $0$ if that slot holds an ad or falls outside the feed.

\begin{table}[h]
\centering
\begin{tabular}{ll}
\toprule
\textbf{Externality function} & \textbf{Parameter values} \\
\midrule
Additive Similarity   & $\eta = 0.5$, $\lambda = 0.7$, $w = 4$ \\
Neighborhood Activation & $\eta = 0.5$, $\tau_{\mathrm{act}} = 0.3$, $w = 4$ \\
Softmax Attention & $\eta = 0.5$, $\tau_{\mathrm{sm}} = 0.5$, $\lambda = 0.7$, $w = 4$ \\
Multiplicative interactions    & $\eta = 0.5$, $\lambda = 0.7$, $w = 4$ \\
Topic affinity                 & $\eta = 0.5$, $\lambda = 0.7$, $w = 4$, $K = 20$ topics \\
Position-modulated             & $\eta_{\mathrm{above}} = 0.6$, $\eta_{\mathrm{below}} = 0.2$ \\
One side Position-modulated & $\eta_{\mathrm{above}} = 0.5$  (Default Externality)\\
\bottomrule
\end{tabular}
\caption{Default parameter values used in all experiments.}
\label{tab:externality_params}
\end{table}

\section{Proofs From \Cref{sec:sample_complexity}}
\label{app:fullproofs}
In this section, we provide the remaining proof steps to bound the sample complexity of MIRT. 

\begin{lemma} \label{lemma:bottleneck}
For any two parameter vectors $\theta_1, \theta_2 \in \mathbb{R}^W$ and any auction instance $S = (C, B)$, if $R_{\theta_1}(C) = R_{\theta_2}(C)$, then $w(\theta_1; S) = w(\theta_2; S)$.
\end{lemma}

\begin{proof}
Fix $\theta_1,\theta_2\in\mathbb{R}^W$ and an auction instance $S=(C,B)$ such that $R_{\theta_1}(C)=R_{\theta_2}(C)$.

Then
\begin{align*}
w(\theta_1;S) &= \max_{x\in R_{\theta_1}(C)}V(x,B) \\
&= \max_{x\in R_{\theta_2}(C)}V(x,B)\\
&=
w(\theta_2;S).
\end{align*}
The first and third equalities follow from the definition of $w$, while the second follows from the assumption that the two parameter vectors generate the same range on context profile $C$.
\end{proof}

\begin{lemma}[Algorithmic Branching of Beam Search] \label{lemma:comparisons}
Fix a context profile $C$ containing $N$ candidate pieces of content. Generating a feed of length $T$ using beam search with width $K$ evaluates at most $E_{total} = O(T K^2 N^2)$  pairwise comparisons of scores assigned by the transformer to candidate continuations. 
Given a fixed tie-breaking rule, the outcomes of these comparisons uniquely determine the output range $R_\theta(C)$.
\end{lemma}

\begin{proof}
At each generation step $t \in \{1,\dots,T\}$, beam search maintains at most $K$ active prefixes. Each active prefix has at most $N$ possible continuations, so there are at most $KN$ candidate continuations at that step.

Let 
$s_\theta(C(x_{t}),p) \in \mathbb{R}$ 
denote the score assigned by beam search at step $t$ to a candidate continuation $p$ under context $C(x_t)$. 
Conditional on the beam state entering at step $t$ (i.e., the set of current beams 
$\{x_{t-1}^1, \dots, x_{t-1}^K \}$), the set of candidate continuations is fixed. The top $K$ continuations are therefore uniquely determined by the outcomes of all pairwise comparisons 
$s_\theta(C(x_{t-1}^k),p) > s_\theta(C(x_{t-1}^{k'}),p')$, $k, k' \in [K], p, p' \in [N]$
among these candidate continuations, together with the fixed tie-breaking rule. 
The number of such comparisons is at most 
$q = \binom{KN}{2} = O(K^2N^2).$

Thus, conditional on the beam state entering step $t$, the outcomes of at most $q$ pairwise score comparisons uniquely determine the beam state entering step $t+1$, applying this argument sequentially over the $T$ generation steps shows that a complete execution evaluates at most
\begin{equation}
E_{\mathrm{total}} =  Tq = T {KN \choose 2} =  O(TK^2N^2)
\end{equation}
comparisons and that their sequential outcomes uniquely determine the final output range $R_\theta(C)$.
\end{proof}

\begin{lemma}[Growth of Beam-Search Ranges] \label{lemma:growth} Let $\mathcal{H}_{\mathrm{cmp}}$ be the binary hypothesis class induced by pairwise comparisons of scores assigned by the transformer class,
$d_{\mathcal{H}_\mathrm{cmp}} = \operatorname{VCdim} \left( \mathcal{H}_{\mathrm{cmp}} \right)$ be its VC-dimension and let $q=\binom{KN}{2}$. 
For any fixed collection of $m$ context profiles $C_1,\dots,C_m$ arising from $m$ auctions, the number of distinct range tuples $\left(R_\theta(C_1),\dots,R_\theta(C_m)\right)$ that can be generated as $\theta$ varies is at most 
\begin{equation} \left( \frac{em q}{d_{\mathcal{H}_\mathrm{cmp}}} \right)^{d_{\mathcal{H}_\mathrm{cmp}}T}, 
\end{equation} if $mq\geq d_{\mathcal{H}_\mathrm{cmp}}$, and $2^{mqT}$ otherwise. 
\end{lemma}

\begin{proof}
For each $t\in\{0,\dots,T\}$, let $\Gamma_t(m)$ denote the maximum number of distinct tuples of beam states after $t$ generation steps that can be produced on the fixed contexts $C_1,\dots,C_m$ as $\theta$ varies. Before generation begins, the beam states are fixed, and hence $\Gamma_0(m)=1$.

Fix an advertiser context profile $C$, and a realizable set of beam states after step $t-1$. Conditional on these, the candidate continuations considered at step $t$ are fixed. 
By \Cref{lemma:comparisons}, determining the next beam requires at most $q = O(K^2 N^2)$ pairwise comparisons for the fixed context $C$. 
Thus, across the $m$ contexts, conditioned on the beam states after step $t-1$, the subsequent beam states at step $t$ are determined by the labels assigned by $\mathcal{H}_{\mathrm{cmp}}$ to a fixed collection of at most $mq$ comparison inputs.

The number of possible next beam-state tuples conditional on the fixed history is at most $\Pi_{\mathcal{H}_{\mathrm{cmp}}}(mq),$
where $\Pi_{\mathcal{H}_{\mathrm{cmp}}}$ denotes the growth function of $\mathcal{H}_{\mathrm{cmp}}$, i.e., the maximum number of binary labelings that the class $\mathcal{H}_{\mathrm{cmp}}$ can realize on any fixed set of $mq$ comparison queries. 
Conditioning on the fixed beam-state history restricts $\theta$ to a subset of $\mathbb{R}^W$, and restricting the parameter space cannot increase the number of labelings realizable by $\mathcal{H}_{\mathrm{cmp}}$. Therefore, the growth function of the full class remains a valid upper bound.
It follows that
\begin{equation*}
\Gamma_t(m)
\leq
\Gamma_{t-1}(m)
\Pi_{\mathcal{H}_{\mathrm{cmp}}}(mq).
\end{equation*}

Applying this inequality recursively over the $T$ generation steps yields
\begin{equation*}
\Gamma_T(m)
\leq
\left(
\Pi_{\mathcal{H}_{\mathrm{cmp}}}(mq)
\right)^T.
\end{equation*}

However, by Sauer's lemma \citep{SAUER1972145}, when $mq \ge d_{\mathcal{H}_\mathrm{cmp}}$,
we can bound the growth function of $\mathcal{H}_{\mathrm{cmp}}$ as
$\Pi_{\mathcal{H}_{\mathrm{cmp}}}(mq)
\leq
\left(
\frac{e m q}{d_{\mathcal{H}_\mathrm{cmp}}}
\right)^{d_{\mathcal{H}_\mathrm{cmp}}}.$
Therefore,
\begin{equation}
\Gamma_T(m)
\leq
\left(
\frac{e m q}{d_{\mathcal{H}_\mathrm{cmp}}}
\right)^{d_{\mathcal{H}_\mathrm{cmp}}T}.
\end{equation}

If instead $mq < d_{\mathcal{H}_{\mathrm{cmp}}}$, then trivially
% \begin{equation}
$\Pi_{\mathcal{H}_{\mathrm{cmp}}}(mq)
\leq
2^{mq},$
% \end{equation}
since there are at most $2^{mq}$ binary labelings of $mq$ comparison inputs. Hence,
\begin{equation}
\Gamma_T(m)
\leq
2^{mqT}.
\end{equation}

The final output range is determined by the beam state after the $T$-th generation step. Thus, $\Gamma_T(m)$ upper-bounds the number of distinct range tuples
$\left(
R_\theta(C_1),\dots,R_\theta(C_m)
\right),$
which proves the result.
\end{proof}

Now we have all the lemmata necessary for the proof of \Cref{thm:pdim}:

\begin{proof}[Proof of \Cref{thm:pdim}]
Suppose that $\mathcal{W}$ pseudo-shatters $m$ auction instances
$S_1,\dots,S_m$ with corresponding thresholds $y_1,\dots,y_m\in\mathbb{R}$.
Then, as $\theta$ varies, $\mathcal{W}$ realizes all $2^m$ binary labelings of the conditions
$w(\theta;S_i)>y_i, i\in \{1,\dots,m\}.$

By \Cref{lemma:bottleneck}, two parameter vectors that generate the same tuple of ranges on $S_1,\dots,S_m$ generate the same welfare vector, and thus the same binary labeling with respect to the thresholds $y_1,\dots,y_m$. 
Consequently, pseudo-shattering the $m$ instances requires the transformer to generate at least $2^m$ distinct range tuples as its parameters vary in $\mathbb{R}^W$.

Let
$d_{\mathcal{H}_\mathrm{cmp}} = \operatorname{VCdim}(\mathcal{H}_{\mathrm{cmp}})$
and $q = \binom{KN}{2}$. We distinguish the two cases of \Cref{lemma:growth}.

First, suppose that
$m q < d_{\mathcal{H}_\mathrm{cmp}}$.
Then directly
$m
<
\frac{d_{\mathcal{H}_\mathrm{cmp}}}{q}
\leq
d_{\mathcal{H}_\mathrm{cmp}}.$

Now suppose that
$m q \geq d_{\mathcal{H}_\mathrm{cmp}}$.
By \Cref{lemma:growth}, the number of distinct range tuples is at most
$\left(
\frac{e m q}{d_{\mathcal{H}_\mathrm{cmp}}}
\right)^{d_{\mathcal{H}_\mathrm{cmp}}T}.$
Since pseudo-shattering requires at least $2^m$ distinct range tuples, we must have
\begin{equation*}
2^m
\leq
\left(
\frac{e m q}{d_{\mathcal{H}_\mathrm{cmp}}}
\right)^{d_{\mathcal{H}_\mathrm{cmp}}T}.
\end{equation*}

Taking logarithms
\begin{equation*}
m
\leq
d_{\mathcal{H}_\mathrm{cmp}}T
\log_2
\left(
\frac{e m q}{d_{\mathcal{H}_\mathrm{cmp}}}
\right).
\end{equation*}

Using the inversion $m\leq a\log m+b \implies m = O(a \log a + b)$, 
yields
$m = O\left( d_{\mathcal{H}_\mathrm{cmp}}T\log(Tq) \right).$

Since
$q=\binom{KN}{2}=O(K^2N^2)$,
we have
$\log(Tq) = O\left( \log(TKN) \right).$
Hence,
\begin{equation}
m
=
O\left(
d_{\mathcal{H}_\mathrm{cmp}}T\log(TKN)
\right).
\end{equation}

The case $mq<d_{\mathcal{H}_\mathrm{cmp}}$ yields the smaller bound
$m<d_{\mathcal{H}_\mathrm{cmp}}/q$, and is therefore covered by the same asymptotic upper bound. Thus,
\begin{equation}
\operatorname{Pdim}(\mathcal{W}) = 
O\left(
d_{\mathcal{H}_\mathrm{cmp}}T\log(TKN)
\right).
\end{equation}

Finally, by \Cref{thm:pairwise-transformer-vc}, under $L\leq W$,
$d_{\mathcal{H}_{\mathrm{cmp}}}
=
O\left(
WL\log(TW)
\right).$
Therefore,
\begin{equation}
\operatorname{Pdim}(\mathcal{W}) = 
O\left(
WLT
\log(TW)
\log(TKN)
\right).
\end{equation}

\end{proof}

\begin{theorem}[VC dimension of pairwise Transformer comparisons]
\label{thm:pairwise-transformer-vc}

Fix a Transformer architecture $A$ of depth $L$ and with $W$ trainable
parameters, as in \citet{yang2026tightcomplexity}, and let
$s_\theta(H)\in\mathbb{R}$
denote its real-valued pre-threshold output on an input sequence
$H\in\mathbb{R}^{T\times d}$.
Define the corresponding pairwise-comparison class
\begin{equation}
\mathcal H_{\mathrm{cmp}}
:=
\left\{
(H,H')
\mapsto
\mathbbm{1}
\left\{
s_\theta(H)>s_\theta(H')
\right\}
:
\theta\in\mathbb{R}^W
\right\}.
\end{equation}
Then
\[
\operatorname{VCdim}(\mathcal H_{\mathrm{cmp}})
=
O\!\left(WL\log(TWL)\right).
\]
In particular, under the assumption $L\le W$,
\begin{equation}
\operatorname{VCdim}(\mathcal H_{\mathrm{cmp}})
=
O\!\left(WL\log(TW)\right).
\end{equation}
\end{theorem}

We prove \Cref{thm:pairwise-transformer-vc} following the parameter-space partitioning argument of \citet{yang2026tightcomplexity}.  
The proof has two parts. 
First, for any fixed collection of input pairs, we partition the parameter space into regions on which all attention-routing and ReLU (or any piecewise-linear activation function) branching decisions are fixed.  
Second, we show that within each such region, each pairwise comparison is the sign of a polynomial in the Transformer parameters of a bounded degree, allowing us to bound the number of comparison patterns using Warren's theorem \citep{AnthonyBartlett1999}.

\begin{lemma}[Partition of the parameter space into polynomial regions]
\label{lem:pairwise-partition}
Fix $m$ arbitrary input pairs
$(H_1^+,H_1^-),\ldots,(H_m^+,H_m^-)
\in
\left(\mathbb{R}^{T\times d}\right)^2.$
There exists a partition $\mathcal P_L$ of $\mathbb{R}^W$ such that, for
every region $P \in\mathcal P_L$:

\begin{enumerate}
    \item all attention-routing decisions and ReLU activation patterns of
    The Transformers are fixed on each of the $2m$ inputs
    $H_1^+,H_1^-,\ldots,H_m^+,H_m^-;$
    \item for every $i\in[m]$, both
    $s_\theta(H_i^+)$ and $s_\theta(H_i^-)$ coincide on $P$ with
    polynomials in $\theta$ of degree at most
    $D_{\mathrm{out}}=O(L);$
    \item the number of regions satisfies
    $\log |\mathcal P_L|
    =
    O\!\left(
    WL(\log m+\log T+\log L)
    \right).$
\end{enumerate}
\end{lemma}

\begin{proof}
The proof follows the recursive parameter-space partitioning argument of
\citet[Appendix~B.1]{yang2026tightcomplexity}. For any collection of $I$
fixed input sequences, their construction partitions the parameter space
$\mathbb R^W$ into regions $\mathcal P_L$ such that, within every
$P \in\mathcal P_L$, all attention-routing decisions and ReLU activation
patterns on all $I$ inputs are fixed. Consequently, the Transformer output
on each input coincides, within $P$, with a polynomial in the parameters
$\theta$ of degree $O(L)$. 
They upper bound the number of such regions as: 
\begin{equation}
\log |\mathcal P_L|
=
O\!\left(
WL(\log I+\log T+\log L)
\right);
\end{equation}
see Appendix~B.1, in particular equations~(28)--(30), of
\citet{yang2026tightcomplexity}.

We apply this construction to the $2m$ fixed input sequences
$H_1^+,H_1^-,\ldots,H_m^+,H_m^-$. The construction applies to an arbitrary
collection of fixed inputs, so no relationship between $H_i^+$ and $H_i^-$,
or between different pairs, is required. Thus, we simply set $I=2m$ in
the bound above. It follows that, within every region $P \in\mathcal P_L$,
all branching decisions are fixed simultaneously on all $2m$ inputs, and,
for every $i\in[m]$, both $s_\theta(H_i^+)$ and $s_\theta(H_i^-)$ are
polynomials in $\theta$ of degree $D_{\mathrm{out}}=O(L)$. Moreover,
\begin{align}
\log |\mathcal P_L|
&=
O\!\left(
WL(\log(2m)+\log T+\log L)
\right) \nonumber \\
&=
O\!\left(
WL(\log m+\log T+\log L)
\right).
\end{align}
\end{proof}

\begin{lemma}[Comparison patterns within a polynomial region]
\label{lem:pairwise-local-patterns}
Fix a region $P\in\mathcal P_L$ as in \Cref{lem:pairwise-partition}.  
For each input pair $i\in[m]$, define 
$g_i^P(\theta) = s_\theta(H_i^+)-s_\theta(H_i^-), \theta\in P.$
Then on $P$, each $g_i^P$  coincides with a polynomial in $\theta$  of degree
at most $D_{\mathrm{out}}=O(L).$ Moreover, if $m\ge W$, the number $\Pi_P^{\mathrm{cmp}}$ of distinct
pairwise-comparison vectors
$\left(
\mathbbm{1}\{s_\theta(H_1^+)>s_\theta(H_1^-)\},
\ldots,
\mathbbm{1}\{s_\theta(H_m^+)>s_\theta(H_m^-)\}
\right)$
realized as $\theta$ varies over $P$ satisfies
$\Pi_P^{\mathrm{cmp}}
\le
2
\left(
\frac{2emD_{\mathrm{out}}}{W}
\right)^W.$
\end{lemma}

\begin{proof}
By \Cref{lem:pairwise-partition}, for each $i\in[m]$ there exist
polynomials $p_{i,+}^P$ and $p_{i,-}^P$ such that, on $P$,
$s_\theta(H_i^+)=p_{i,+}^P(\theta), s_\theta(H_i^-)=p_{i,-}^P(\theta),$
with
$\deg(p_{i,+}^P) \le D_{\mathrm{out}}$ and 
$\deg(p_{i,-}^P) \le D_{\mathrm{out}}.$
Thus
$g_i^P(\theta)
=
p_{i,+}^P(\theta)-p_{i,-}^P(\theta)$
is itself a polynomial in the same $W$ parameters, of degree at most 
\[
\deg(g_i^P)
\le
\max\{
\deg(p_{i,+}^P),
\deg(p_{i,-}^P)
\}
\le
D_{\mathrm{out}}.
\]
Although $2m$ transformer executions are required to define
the partition, there are only $m$ output-comparison polynomials
$g_1^P,\ldots,g_m^P$.
Furthermore, $\mathbbm{1}\{s_\theta(H_i^+)>s_\theta(H_i^-)\} = \mathbbm{1}\{g_i^P(\theta)>0\}.$
Thus the entire comparison vector is determined by the ternary sign pattern
of the $m$ polynomials $g_1^P,\ldots,g_m^P.$
Applying Warren's theorem to these $m$ polynomials in $W$ variables, each
of degree at most $D_{\mathrm{out}}$, gives
\[
\Pi_P^{\mathrm{cmp}}
\le
2
\left(
\frac{2emD_{\mathrm{out}}}{W}
\right)^W.
\]
\end{proof}

\begin{proof}[Proof of Theorem~\ref{thm:pairwise-transformer-vc}]
Suppose that $\mathcal H_{\mathrm{cmp}}$  shatters $m$ input pairs $(H_1^+,H_1^-),\ldots,(H_m^+,H_m^-).$
If $m<W$, then
$m<W
=
O(WL\log(TWL)),
$
and the claimed bound follows immediately.  Hence assume $m\ge W$.

Let $\mathcal P_L$ be the partition supplied by
Lemma~\ref{lem:pairwise-partition}.  Since the $m$ pairs are shattered,
all $2^m$ binary comparison vectors are realized.  
Decomposing according to the partition yields:
\[
2^m
\le
\sum_{P\in\mathcal P_L}
\Pi_P^{\mathrm{cmp}}
\le
|\mathcal P_L|
\max_{P\in\mathcal P_L}
\Pi_P^{\mathrm{cmp}}.
\]
Applying Lemmata ~\ref{lem:pairwise-partition} and ~\ref{lem:pairwise-local-patterns}, we have 
$2^m \le |\mathcal P_L| \cdot 2 \left(
\frac{2emD_{\mathrm{out}}}{W}
\right)^W,$
where
$\log|\mathcal P_L| = O\!\left( WL(\log m+\log T+\log L) \right)$
and
$D_{\mathrm{out}}=O(L).$
Taking logarithms yields
\begin{align*}
m
& = 
O\!\left(
WL(\log m+\log T+\log L)
\right)
+
O\! \left (W\log\left(
\frac{2 e m D_{\mathrm{out}}}{W}
\right) \right )\\
&=
O\!\left(
WL\log m
+
WL\log T
+
WL\log L  
\right) + O \left (W \log \left (\frac{mL}{W} \right ) \right ) \\ 
&=
O\!\left(
WL\log m
+
WL\log T
+
WL\log L
\right).
\end{align*}
% where the final term is absorbed into the first-order $WL$ terms since $L\ge1$.

Thus
$m \le O(WL\log m ) + O\!\left(WL(\log T+\log L)\right).$
Using the inversion
$m\le a\log m+b
\implies
m=O(a\log a+b),$ \citep[Lemma A.2]{understandingML}
we obtain
\begin{align*}
m
&=
O\!\left(
WL\log(WL)
+
WL\log T
+
WL\log L
\right) \\
&=
O\!\left(
WL\log(TWL)
\right).
\end{align*}
Therefore,
\begin{equation}
    \operatorname{VCdim}(\mathcal H_{\mathrm{cmp}}) = O\!\left(WL\log(TWL)\right).    
\end{equation}

Finally, under the assumption that the transformer has fewer blocks $L$ than trainable parameters $W$, 
we have $\log(TWL) \le \log(TW^2)=O(\log(TW)),$
thus
\begin{equation}
\operatorname{VCdim}(\mathcal H_{\mathrm{cmp}})
=
O\!\left(WL\log(TW)\right).    
\end{equation}

\end{proof}

\begin{lemma}[Warren's Bound, Lemma 2.1 in \citet{AnthonyBartlett1999}]
 Let $p_1, p_2, \dots p_m$ be polynomials in $W$ variables, each of degree at most $D$. If $W \le m$, the number of distinct sign vectors 
 $(sgn(p_1(\theta)), sgn(p_2 (\theta)), \dots, sgn(p_m (\theta))) \in \{-1, 0, 1 \}^m$ 
 as $\theta$ varies in $\mathbb{R}^W$ is at most $2(2emD/W)^W.$
\end{lemma}

\section{Environment Setup}
\label{sec:env_setup}
We construct a synthetic environment with large pools of ads and organic posts, whose embeddings are sampled from Gaussian clusters around randomly generated topic centroids. Each user is represented by a weighted average of these centroids, with user-specific weights capturing different interests across topics. 

Each instance of a dataset is constructed by uniformly sampling $M$ ads and $N_O$ organic posts without replacement from their respective pools for a given user embedding. We then sample the ad qualities, bids, and organic benefits for the selected user--item pairs. The resulting item embeddings, ad qualities, and organic benefits form the context $C$ supplied to the range generator.

To model user interests, ad qualities, and organic benefits, samples are drawn based on user-item similarities, with their distributions determined by the cosine similarity between the user's representation and the item's embedding. Ad qualities and organic post benefits follow lognormal/Weibull distributions characterized by the user-item similarities. However, these distributions and their noise levels are chosen so that the corresponding advertiser and user welfare distributions remain balanced (Appendix Figure \ref{fig:welfare_distributions}). Thus, items that are more closely aligned with a user's interests tend to have higher qualities or benefits, while the two components of welfare stay comparable. Advertiser benefits-per-click are sampled i.i.d. from an exponential or other distribution, once per ad when the global pool is created, and are also sampled fresh during training. Figure~\ref{fig:welfare_distributions} summarizes the resulting item benefits and feed-level welfare, including the variation across alternative feeds for the same user and candidate pool.

Engagement rates depend on an item's intrinsic quality, position, and surrounding content through the permutation-externality models defined in Section~\ref{sec:model_and_preliminaries}.

\begin{figure}[tbp]
  \centering
  \caption{ \textbf{Welfare distributions for the default environment.} Generated over 2,000 contexts with 256 random feasible feeds each. From left to right: (1) ad bid--quality products and organic benefits; (2) welfare per feed; (3) within-context standard deviation across feeds; (4) within-context maximum minus minimum among sampled feeds. Blue denotes advertiser quantities and green denotes organic quantities.
  }
  \includegraphics[width=\linewidth]{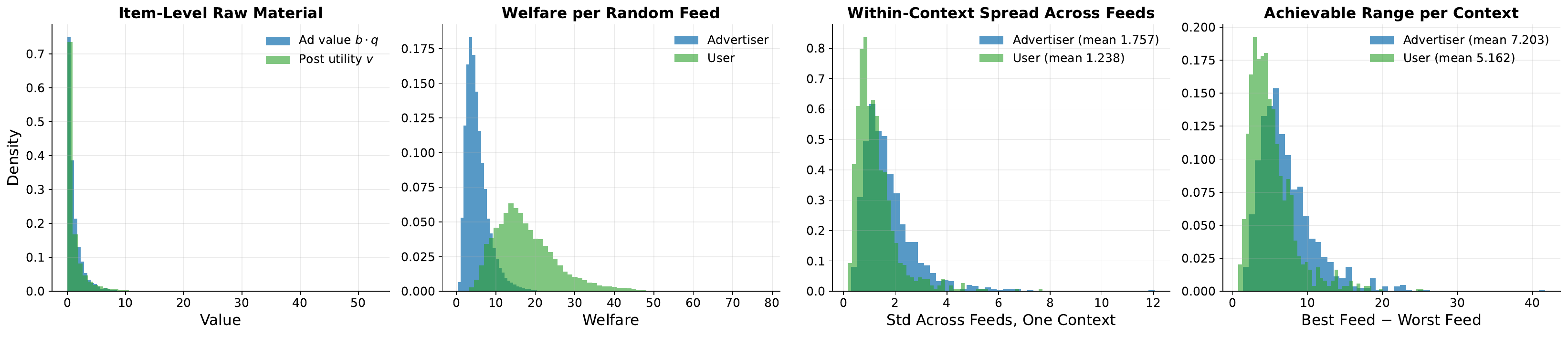}
  
  \label{fig:welfare_distributions}
\end{figure}

\section{Further MIRT Training Details}
\label{app:mirttrainingapp}

Our MIRT networks have 1.38M parameters, with four encoder and two decoder layers, each with eight attention heads and a 128-dimensional hidden size.

We optimize with AdamW \citep{loshchilov2017decoupled}, using an initial learning rate of $10^{-4}$, weight decay of $10^{-4}$, and gradient-norm clipping at $1.0$. The learning rate follows a cosine decay schedule without restarts \citep{loshchilov2017sgdr}, over a maximum of 40 epochs with early stopping when there is no more than 0.3 percent improvement in the last 12 epochs. Each optimizer update uses 128 contexts, processed in two micro-batches with gradient accumulation.
\begin{figure}[tbp]
    \centering
    \caption{This figure shows the variation in performance of optimizing the softmax sampling temperature at test time on frozen models. For models trained on welfare, a higher temperature results in decreased welfare and increased total transfer.}
    \includegraphics[width=0.7\linewidth]{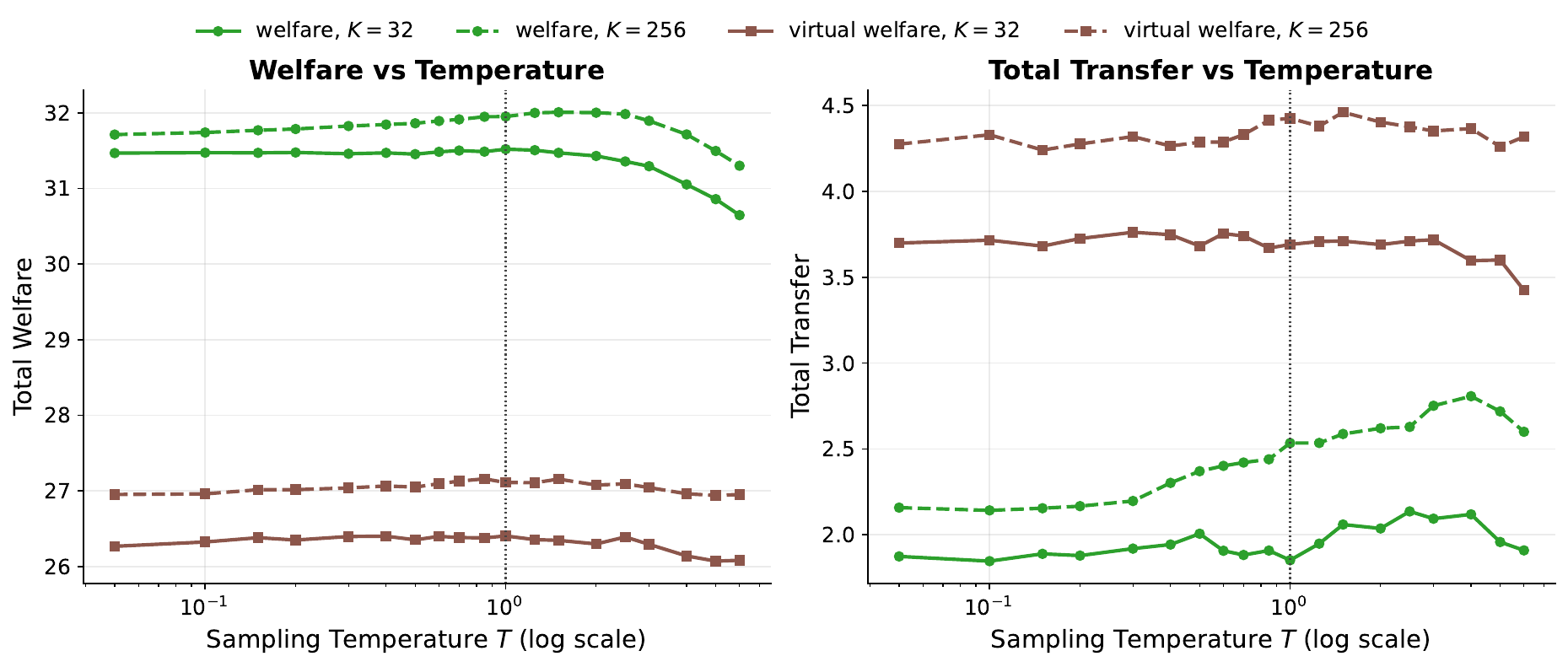}
    % \label{fig:topline_all}
    \label{fig:temperatureanalysis}
\end{figure}

\section{CGA and NGA training details}
\label{app:cgangatraining}
\subsection{Contextual Generative Auctions}
CGA~\citep{cga2024} is an ad auction mechanism in which each context consists of a user’s features, candidate ad features, and their bids, and predicted click-through rates for that user-ad pair. Each context has 30 candidate ads that are competing for three slots. The architecture consists of three separately trained networks, an evaluator, a generator, and the PaymentNet. 

The evaluator receives a historically shown ad sequence, user features, and existing individual-ad click predictions and is trained on historical click data to predict each ad’s click-through rate. It processes the sequence using a bidirectional LSTM and combines its outputs with user features through an MLP to predict an externality multiplier for each displayed ad. This multiplier is applied to the existing ad-quality estimate to obtain the adjusted click-through rate. Unlike ours, their externality model captures only interactions among ads.

The generator uses a self-attention encoder to represent the candidate ads and a GRU decoder to select them one at a time. The frozen evaluator scores the generated sequences, providing rewards for policy-gradient training that aims to maximize the virtual welfare of a three-ad sequence. Finally, PaymentNet receives the selected ads’ representations, bid information, and features derived from allocation probabilities and predicted click-through rates. It predicts a transfer fraction between zero and one, which is multiplied by the advertiser’s bid to determine its transfer per click. PaymentNet is trained to increase total transfer while penalizing ex-post regret, which measures how much advertisers could benefit by misreporting their bids. A tunable hyperparameter controls the strength of this penalty.

\paragraph{CGA ad-organic implementation.}
The original CGA paper considers only ads. To create a comparable baseline, we extend its candidate pool and architecture to include organic posts. In our implementation, the evaluator predicts each ad’s click-through rate after accounting for position and surrounding content, and each organic post’s welfare contribution after accounting for position. It is trained using our environment instead of historical data and learns to approximate the externality model used by the welfare function. In the generator, we added separate input projections and learned item-type embeddings to map both ad and organic types into a shared representation space, so the CGA generator's self-attention is over the entire candidate pool. Finally, similar to the original paper, PaymentNet is trained on synthetic auction contexts with the generator and evaluator frozen. It predicts each displayed advertiser’s transfer per click as a fraction of its bid, to maximize total transfer while penalizing ex-post regret. To estimate regret, we change an advertiser’s bid, regenerate the feed, and recompute transfers and click predictions to measure whether the advertiser benefits from misreporting. 

Our CGA implementation has approximately 203,000 parameters and uses a two-layer Transformer encoder with hidden dimension 64, four attention heads, and feed-forward width 256, a 64-dimensional GRU decoder, a bidirectional LSTM evaluator with attention and separate ad and organic prediction heads, and a PaymentNet is a four layer MLP with 13,000 training parameters.

\subsection{Non-autoregressive Generative Auctions}
NGA~\citep{nga2025} has a two-network architecture, a generator for allocation, and an evaluator to score allocations and calculate transfers. It extends the CGA architecture to also consider static surrounding organic content in its ad allocation. Unlike CGA's autoregressive GRU decoder, NGA uses item and position encoders to compute allocation scores for all item-position pairs, then creates candidate feeds item by item through a decoder non-autoregressively. The evaluator scores generator's feed by producing click-through rates, conversion rates, and transfer fractions per ad. These predictions are used to calculate a score based on the platform's total transfer for a feed.

\paragraph{NGA implementation.}
In our setting, organic posts are selected and reordered alongside ads, so we redesigned NGA's generator to encode ad-organic context and pass the shared item representations through an MLP to generate a feed. The evaluator is the same simulator-trained network we use for CGA, with separate heads for ad click-through rates and organic welfare contributions. In the welfare experiments, we train the generator with policy gradients on the combined welfare. Transfers come from the separately trained head of the evaluator that predicts transfer fraction to maximize virtual welfare and penalize regret. 

Our NGA implementation has approximately 180,000 parameters and uses the same encoder, evaluator, and a four layer transfer head, replacing the GRU decoder with learned 64-dimensional position embeddings and a three layer MLP that scores all item–position pairs in parallel in the generator.

\subsection{Regret}
Regret is the largest increase in an advertiser's utility that misreporting its bid can produce. We estimate it for each advertiser selected under truthful bidding. We iterate over a grid of 20 bids for an advertiser holding the other bids fixed, and recompute the allocation and transfers. We sum the positive gains across advertisers and normalize by the advertisers' truthful utility. During training, the regret penalty $\rho$ is $1$ unless otherwise noted.

\section{Additional Experimental Information}
\label{app:additionalexperiments}

\subsection{Virtual Welfare Experiments}
\label{app:virtualwelfare}

Our mechanism design setting is quite challenging, as the combinatorial structure of feeds makes optimization difficult. However, because advertiser welfares are click-through rates multiplied by scalars, it is a single-parameter setting, and therefore relatively well-behaved compared to general multi-parameter mechanism design~\citep{archer2001truthful,Myerson81}. In particular, for our setting, assuming that CTR is nondecreasing in a bidder's bid, and that the lowest bidder type receives 0 utility, the corresponding Myerson transfer is $t_i(B) = b_i q_i(x(B)) - \int_0^{b_i} q_i(x(z,B_{-i}))\,dz.$

A particularly useful property in single-parameter settings is that it is possible to define a quantity known as \textit{virtual welfare.} The \textit{virtual value} of a bidder with value $v_i$ is $v_i - \frac{1-F_i(v_i)}{f_i(v_i)}$. For any DSIC mechanism in which the low type has utility 0, expected virtual welfare equals expected total transfer.
Assuming a property known as \textit{regularity} of the distribution, the virtual-welfare-maximizing mechanism 
can be combined with an appropriate rule to maximize total transfer, while also being DSIC. 
For our experiments, we also define a virtual-welfare-optimizing MIR mechanism. Equipped with the appropriate Myerson transfer, this is also DSIC for regular distributions. Results on the default environment are in \Cref{fig:topline_vw}. MIRT competes with other methods on welfare, outperforms baselines on total transfer, and competes with CGA and NGA on total transfer while being exactly strategyproof.
\begin{figure}[tbp]
    \centering
    \caption{Results of training models to optimize a virtual welfare objective instead of welfare. MIR methods use Myerson transfers instead of VCG transfers. MIRT has a sampling temperature of 1, and CGA/NGA have $\rho=1000.$}
    \includegraphics[width=0.98\linewidth]{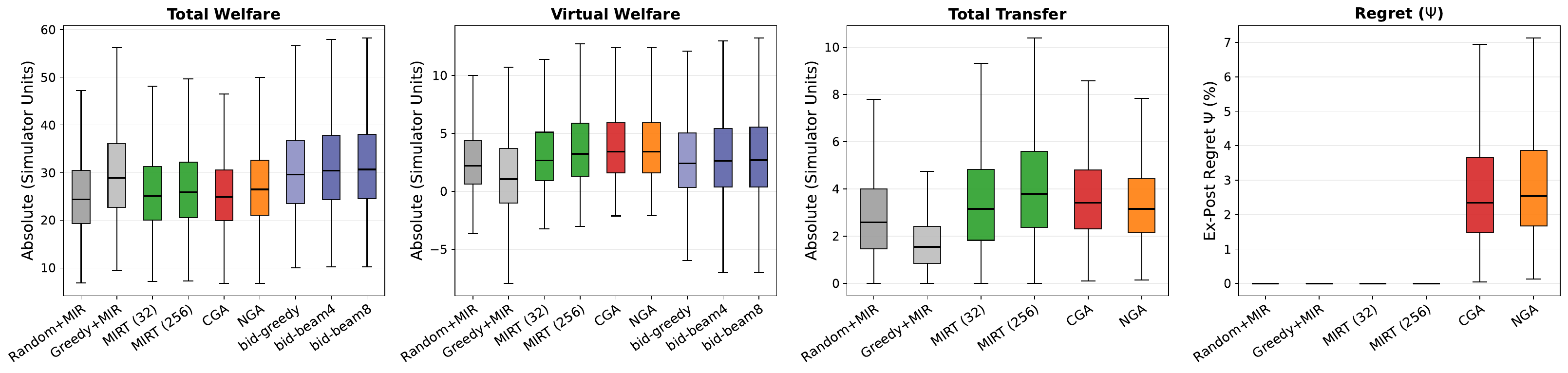}
    \label{fig:topline_vw}
\end{figure}

\subsection{Additional Information from Experiments in \Cref{sec:experiments}}
\begin{table}[htbp]
\centering
\caption{Mean ($\pm$ CI across held-out contexts) absolute total transfer across held-out contexts under each externality function with default problem size and bid distribution. Highlighted cells are the best-performing among mechanisms. Results in this table for all columns were generated using MIRT models trained for 40 epochs with a temperature of 1. CGA/NGA have $\rho=1$, resulting in 10-16\% regret.} %\mjc{confirmed these numbers are right}}
\label{tab:revenuefromorigtable}
\resizebox{\linewidth}{!}{
\begin{tabular}{l *{7}{c}}
\toprule
Method & Prev. Organic & Additive & Activation & Attention & Multiplicative & Topic & Position \\
\midrule
% Method & cols 2--8 of each original row
Random+MIR & $2.01 \pm 0.04$ & $2.06 \pm 0.04$ & $2.74 \pm 0.05$ & $2.36 \pm 0.05$ & $2.13 \pm 0.04$ & $2.00 \pm 0.04$ & $2.08 \pm 0.04$ \\
Greedy+MIR & $1.43 \pm 0.02$ & $1.42 \pm 0.02$ & $1.74 \pm 0.03$ & $1.58 \pm 0.03$ & $1.43 \pm 0.02$ & $1.40 \pm 0.02$ & $1.47 \pm 0.02$ \\
MIRT $(K = 32)$ & $2.13 \pm 0.03$ & $2.06 \pm 0.03$ & $2.25 \pm 0.04$ & $2.03 \pm 0.03$ & $2.28 \pm 0.04$ & $2.20 \pm 0.04$ & $2.22 \pm 0.04$ \\
MIRT $(K = 256)$ & $2.64 \pm 0.04$ & $2.52 \pm 0.04$ & $2.78 \pm 0.04$ & $2.58 \pm 0.04$ & $2.81 \pm 0.04$ & $2.73 \pm 0.04$ & $2.72 \pm 0.04$ \\
CGA & \hlg{$4.93 \pm 0.04$} & $3.89 \pm 0.02$ & $4.20 \pm 0.03$ & $4.51 \pm 0.03$ & \hlg{$5.13 \pm 0.04$} & $4.35 \pm 0.03$ & $4.69 \pm 0.03$ \\
NGA & $4.32 \pm 0.03$ & \hlg{$4.60 \pm 0.03$} & \hlg{$5.24 \pm 0.04$} & \hlg{$5.75 \pm 0.04$} & $4.85 \pm 0.04$ & \hlg{$4.51 \pm 0.03$} & \hlg{$5.12 \pm 0.04$} \\
\bottomrule
\end{tabular}
}
\end{table}

% \end{table}

\subsection{Varying problem sizes}
\label{sec:problem-size-experiments}
We study how the methods scale with the numbers of candidate ads, organic posts, and feed positions. For each problem size, we retain the bid distribution, externality function,  and organic-welfare weight from the reference environment. Results are in \Cref{tab:problem_size}. Our method continues to perform competitively against zero regret baselines in both total transfers and welfare across the problem sizes. 

\begin{table}[htbp]
\centering
\caption{Absolute total welfare and transfer (mean $\pm$ CI across held-out contexts) across different problem sizes, trained under the reference configuration. MIRT used a decoding temperature of 2.4, and CGA/NGA used $\rho=1000.$}
\label{tab:problem_size}
\resizebox{\textwidth}{!}{%
\begin{tabular}{l cc cc cc}
\toprule
& \multicolumn{2}{c}{\textbf{\makecell{$(N_O{=}12,\ M{=}8,\ T{=}12)$\\ \small 3 ads, 9 posts}}}
& \multicolumn{2}{c}{\textbf{\makecell{Default $(N_O{=}18,\ M{=}10,\ T{=}24)$\\ \small 6 ads, 18 posts}}}
& \multicolumn{2}{c}{\textbf{\makecell{$(N_O{=}36,\ M{=}18,\ T{=}36)$\\ \small 9 ads, 27 posts}}} \\
\cmidrule(lr){2-3} \cmidrule(lr){4-5} \cmidrule(lr){6-7}
\textbf{Method}
& \textbf{Welfare} & \textbf{Total Transfer}
& \textbf{Welfare} & \textbf{Total Transfer}
& \textbf{Welfare} & \textbf{Total Transfer} \\
\midrule
Random+MIR
& $17.838 \pm 0.105$ & $2.155 \pm 0.039$
& $27.155 \pm 0.214$ & $2.047 \pm 0.039$
& $35.043 \pm 0.125$ & $2.785 \pm 0.057$ \\
Greedy+MIR
& $20.048 \pm 0.122$ & $1.810 \pm 0.029$
& $30.708 \pm 0.240$ & $1.437 \pm 0.023$
& $44.619 \pm 0.172$ & $2.102 \pm 0.035$ \\
MIRT $(K{=}32)$
& $20.885 \pm 0.128$ & $2.585 \pm 0.037$
& $32.020 \pm 0.244$ & $2.258 \pm 0.036$
& $46.882 \pm 0.182$ & $3.025 \pm 0.044$ \\
MIRT $(K{=}256)$
& \hlg{$21.346 \pm 0.131$} & \hlg{$3.077 \pm 0.040$}
& \hlg{$32.500 \pm 0.246$} & $2.799 \pm 0.037$
& \hlg{$47.904 \pm 0.186$} & $4.017 \pm 0.048$ \\
CGA
& $20.222 \pm 0.117$ & $1.955 \pm 0.013$
& $31.084 \pm 0.238$ & \hlg{$3.125 \pm 0.021$}
& $47.076 \pm 0.172$ & $6.484 \pm 0.050$ \\
NGA
& $20.238 \pm 0.117$ & $1.788 \pm 0.011$
& $31.095 \pm 0.238$ & $2.411 \pm 0.013$
& $47.302 \pm 0.172$ & \hlg{$6.761 \pm 0.046$} \\
\midrule
bid-greedy
& $20.260 \pm 0.117$ & n/a
& $31.121 \pm 0.238$ & n/a
& $47.391 \pm 0.172$ & n/a \\
bid-beam4
& $20.890 \pm 0.123$ & n/a
& $31.956 \pm 0.243$ & n/a
& $49.028 \pm 0.180$ & n/a \\
bid-beam8
& $21.145 \pm 0.127$ & n/a
& $32.205 \pm 0.244$ & n/a
& $49.516 \pm 0.187$ & n/a \\
\bottomrule
\end{tabular}%
}
\end{table}

\begin{table}[htbp]
\centering
\caption{Absolute metrics across methods for the reference configuration as described in \Cref{sec:experiments}. Means with approximate 95\% confidence intervals across held-out test contexts in the test set, alongside medians.
Results in this table were generated with 40 epochs of training for each MIRT model with inference time temperature $2.4$. For CGA/NGA, $\rho=1000$.}
\label{tab:absolute_experiment1}
\resizebox{\textwidth}{!}{%
\begin{tabular}{l cc cc cc cc cc}
\toprule
& \multicolumn{2}{c}{\textbf{Organic Welfare}}
& \multicolumn{2}{c}{\textbf{Advertiser Welfare}}
& \multicolumn{2}{c}{\textbf{Total Welfare}}
& \multicolumn{2}{c}{\textbf{Total Transfer}}
& \multicolumn{2}{c}{\textbf{Regret ($\Psi$)}} \\
\cmidrule(lr){2-3} \cmidrule(lr){4-5} \cmidrule(lr){6-7} \cmidrule(lr){8-9} \cmidrule(lr){10-11}
\textbf{Method} & \textbf{Mean $\pm$ CI} & \textbf{Median} & \textbf{Mean $\pm$ CI} & \textbf{Median} & \textbf{Mean $\pm$ CI} & \textbf{Median} & \textbf{Mean $\pm$ CI} & \textbf{Median} & \textbf{Mean $\pm$ CI} & \textbf{Median} \\
\midrule
Random+MIR   & $18.937 \pm 0.202$ & $17.307$ & $8.217 \pm 0.097$  & $7.362$  & $27.155 \pm 0.214$ & $25.739$ & $2.047 \pm 0.039$ & $1.687$ & $0.000 \pm 0.000$ & $0.000$ \\
Greedy+MIR   & $20.752 \pm 0.217$ & $19.089$ & $9.956 \pm 0.112$  & $8.975$  & $30.708 \pm 0.240$ & $29.177$ & $1.437 \pm 0.023$ & $1.237$ & $0.000 \pm 0.000$ & $0.000$ \\
MIRT (32)    & $20.554 \pm 0.214$ & $18.877$ & $11.466 \pm 0.124$ & $10.412$ & $32.020 \pm 0.244$ & $30.486$ & $2.258 \pm 0.036$ & $1.949$ & $0.000 \pm 0.000$ & $0.000$ \\
MIRT (256)   & $20.562 \pm 0.214$ & $18.910$ & $11.938 \pm 0.128$ & $10.865$ & $32.500 \pm 0.246$ & $30.904$ & $2.799 \pm 0.037$ & $2.491$ & $0.000 \pm 0.000$ & $0.000$ \\
CGA          & $20.915 \pm 0.217$ & $19.239$ & $10.169 \pm 0.112$ & $9.215$  & $31.084 \pm 0.238$ & $29.558$ & $3.125 \pm 0.021$ & $3.006$ & $0.105 \pm 0.002$ & $0.085$ \\
NGA          & $20.874 \pm 0.217$ & $19.199$ & $10.221 \pm 0.112$ & $9.278$  & $31.095 \pm 0.238$ & $29.545$ & $2.411 \pm 0.013$ & $2.349$ & $0.112 \pm 0.002$ & $0.091$ \\
bid-greedy   & $20.863 \pm 0.217$ & $19.196$ & $10.257 \pm 0.112$ & $9.309$  & $31.121 \pm 0.238$ & $29.592$ & \multicolumn{1}{c}{n/a} & \multicolumn{1}{c}{n/a} & \multicolumn{1}{c}{n/a} & \multicolumn{1}{c}{n/a} \\
bid-beam4    & $20.847 \pm 0.216$ & $19.165$ & $11.109 \pm 0.119$ & $10.160$ & $31.956 \pm 0.243$ & $30.421$ & \multicolumn{1}{c}{n/a} & \multicolumn{1}{c}{n/a} & \multicolumn{1}{c}{n/a} & \multicolumn{1}{c}{n/a} \\
bid-beam8    & $20.830 \pm 0.216$ & $19.151$ & $11.374 \pm 0.122$ & $10.395$ & $32.205 \pm 0.244$ & $30.668$ & \multicolumn{1}{c}{n/a} & \multicolumn{1}{c}{n/a} & \multicolumn{1}{c}{n/a} & \multicolumn{1}{c}{n/a} \\
\bottomrule
\end{tabular}%
}
\end{table}